%% file: ms.tex
\documentclass{article}
\pdfoutput=1
\usepackage{verbatim}
\usepackage[T1]{fontenc}
\usepackage[utf8]{inputenc}\DeclareUnicodeCharacter{2212}{-}
\usepackage{verbatim}
\usepackage{mathrsfs}
\usepackage{amsmath}
\usepackage{amsthm}
\usepackage{amssymb}
\usepackage[mathscr]{euscript}
\usepackage{mathtools}
\usepackage{geometry}
\usepackage{xspace}
\usepackage{blindtext}
\usepackage{relsize}
\usepackage{caption}
\usepackage{subcaption}
\usepackage{float}
\usepackage{pgf}

\usepackage{algorithm}
\usepackage{algpseudocode}
\usepackage{algorithmicx}
\usepackage{etoolbox}
\usepackage{appendix}

\DeclareCaptionType[fileext=ext]{supfigure}[Supplemental Figure]
\DeclareCaptionType[fileext=ext]{supcaption}[Supplemental Figure]
\DeclareCaptionType[fileext=ext]{figcaption}[Figure]
\newdimen\figrasterwd
\figrasterwd\textwidth

\usepackage{tikz}
\usepackage[colorlinks=true]{hyperref}
\usepackage[nameinlink]{cleveref}

\crefformat{onlinemeth}{#2#1#3}
\newcounter{onlinemeth}

\makeatletter

\newcommand{\onlinemeth}[1]{%
  \section*{#1}%
  \refstepcounter{onlinemeth}%
  \def\cref@currentlabel{[onlinemeth][\arabic{onlinemeth}][]#1}%
  \def\@currentlabelname{#1}%
}

\makeatother

\newtheorem{theorem}{Theorem}[section]

\newtheorem{claim}[theorem]{Claim}
\newtheorem{definition}{Definition}[section]

\newtheorem*{remark*}{Remark}

\newtheorem{property}[theorem]{Property}

\crefname{proof}{proof}{proofs}
\Crefname{proof}{Proof}{Proofs}

\crefname{property}{property}{properties}
\Crefname{property}{Property}{Properties}

\crefname{suppmethod}{Supplemental Methods}{Supplemental Methods}
\crefname{methods}{Methods}{Methods}
\Crefname{methods}{Methods}{Methods}
\crefname{sfig}{S.}{S.}
\crefname{cap}{Fig.}{Fig.}

\usepackage{catchfilebetweentags}
\usepackage{xargs}                      
\usepackage[colorinlistoftodos,prependcaption,textsize=tiny]{todonotes}

\usepackage{babel}

\title{Dimensionality Reduction of Longitudinal 'Omics Data using Modern Tensor Factorization}

\author{
    Uria Mor$^{1,2}$\\
    \and 
    Yotam Cohen$^{1}$\\
    \and 
    Rafael Valdés-Mas$^{1}$\\
    \and
    Denise Kviatcovsky$^{1}$\\
    \and
    Eran Elinav$^{1,3,\#}$\\
    \and
    Haim Avron$^{2,\#}$
    }

\date{}

\begin{document}
\input{macros.tex}
\maketitle

\noindent
$^1$Immunology Department, Weizmann Institute of Science, Rehovot, 7610001, Israel.\\%
$^2$School of Mathematical Sciences, Tel Aviv University, Tel Aviv 6997801, Israel.\\%
$^3$Division of Cancer-Microbiome Research, DKFZ, Heidelberg, Germany.\\%
$^{\#}$Equal last contributors.%

\vspace{2cm}

\noindent
All correspondence to:\\
Haim Avron, Ph.D.  \\
School of Mathematical Sciences, Tel Aviv University \\
Ramat Aviv, \\
Tel Aviv, Israel, 6997801 \\
(03) 640-8893 (phone) \\
haimav@tauex.tau.ac.il \\
\& \\
Eran Elinav, M.D., Ph.D. \\
Immunology Department, \\
Weizmann Institute of Science, \\
234 Herzl Street, \\
Rehovot, Israel, 7610001 \\
(08) 934-4014 (phone) \\
eran.elinav@weizmann.ac.il \\

\vspace{2cm}
\noindent
Keywords \\
Longitudinal samples, omics data, dimensionality reduction, tensor factorization

\newpage
\section*{Abstract}
Precision medicine is a clinical approach for disease prevention, detection and treatment, which considers each individual’s genetic background, environment and lifestyle. The development of this tailored avenue has been driven by the increased availability of omics methods, large cohorts of temporal samples, and their integration with clinical data. Despite the immense progression, existing computational methods for data analysis fail to provide appropriate solutions for this complex, high-dimensional and longitudinal data. In this work we have developed a new method termed TCAM, a dimensionality reduction technique for multi-way data, that overcomes major limitations when doing trajectory analysis of longitudinal omics data. Using real-world data, we show that TCAM outperforms traditional methods, as well as state-of-the-art tensor-based approaches for longitudinal microbiome data analysis. Moreover, we demonstrate the versatility of TCAM by applying it to several different omics datasets, and the applicability of it as a drop-in replacement within straightforward ML tasks.


\input{./child_docs/clean_text.tex}

\pagebreak

\onlinemeth{Methods}%
    \label{sec:online.methods}

\input{./child_docs/online_methods}

\vfill
\rule[0.5ex]{1\columnwidth}{1pt}

\pagebreak
\section{Supplementary Figures}
    \input{child_docs/figure_panels/S1.tex}
    \input{./child_docs/figure_panels/S2_IBD.tex}

\pagebreak

\begin{appendices}
\onlinemeth{Supplementary Discussion}%
    \label{app:discussion}
    \input{./child_docs/supp_discussion}

    \section{PCA}\label[suppmethod]{app:matrix.pca.sec}
        \input{./child_docs/supp_pca}

	    \paragraph{Singular Value Decomposition.}
	    \input{./child_docs/supp_pca_svd}

    \section{Tensors component analysis}
        \input{./child_docs/supp_tcam}

\end{appendices}

    \bibliographystyle{plain}
    \bibliography{bibfile}

\end{document}

%% file: macros.tex
\newcommand{\noun}[1]{\textsc{#1}\xspace}
\newcommand{\bs}[1]{\ensuremath{\boldsymbol{#1}}}
\newcommand{\argmin}{\text{\ensuremath{\arg\min}}}
\newcommand{\argmax}{\text{\ensuremath{\arg\max}}}
\newcommand{\st}{\,\,\,\text{s.t.}\,\,\,}
\newcommand{\xx}{\ensuremath{\times}}
\newcommand{\code}[1]{\codex{#1}}
\newcommand{\pinv}{\ensuremath{+}}
\newcommand{\mpn}{\ensuremath{m \xx p \xx n}}
\newcommand{\pmn}{\ensuremath{p \xx m \xx n}}

\newcommand{\tens}[1]{\ensuremath{\bs{\mathscr #1}}}
\newcommand{\tA}{\tens{A}}
\newcommand{\tAt}{\tA^{\T}}
\newcommand{\btA}{\bar{\tens{A}}}
\newcommand{\ttA}{\tilde{\tens{A}}}
\newcommand{\thA}{\widehat{\tA}}
\newcommand{\thAt}{\thA^{\T}}
\newcommand{\tB}{\tens{B}}
\newcommand{\tBt}{\tB^{\T}}
\newcommand{\thB}{\widehat{\tB}}
\newcommand{\thBt}{\thB^{\T}}
\newcommand{\tR}{\tens{R}}
\newcommand{\thR}{\widehat{\tR}}
\newcommand{\tRt}{\tR^{\T}}
\newcommand{\tM}{\tens{M}}
\newcommand{\thM}{\widehat{\tM}}
\newcommand{\tMt}{\tM^{\T}}

\newcommand{\tV}{\tens{V}}
\newcommand{\tVt}{\tV^{\T}}
\newcommand{\thV}{\widehat{\tV}}
\newcommand{\thVt}{\thV^{\T}}
\newcommand{\tU}{\tens{U}}
\newcommand{\tUt}{\tU^{\T}}
\newcommand{\thU}{\widehat{\tU}}
\newcommand{\thUt}{\thU^{\T}}
\newcommand{\tX}{\tens{X}}
\newcommand{\thX}{\widehat{\tX}}
\newcommand{\tY}{\tens{Y}}
\newcommand{\tYt}{\tY^{\T}}
\newcommand{\thY}{\widehat{\tY}}
\newcommand{\thYt}{\thY^{\T}}

\newcommand{\qprm}{\ensuremath{q^{\prime}}}

\newcommand{\tS}{\tens{S}}
\newcommand{\tSt}{\tS^{\T}}
\newcommand{\thS}{\widehat{\tS}}
\newcommand{\thSt}{\thS^{\T}}
\newcommand{\teJ}{\tens{J}}
\newcommand{\theJ}{\widehat{\teJ}}
\newcommand{\tI}{\tens{I}}
\newcommand{\thI}{\widehat{\tI}}
\newcommand{\tC}{\tens{C}}
\newcommand{\tCt}{\tC^{\T}}
\newcommand{\thC}{\widehat{\tC}}
\newcommand{\thCt}{\thC^{\T}}
\newcommand{\tG}{\tens{G}}
\newcommand{\tGt}{\tG^{\T}}
\newcommand{\thG}{\widehat{\tG}}

\newcommand{\tE}{\tens{E}}
\newcommand{\thE}{\widehat{\tE}}
\newcommand{\tQ}{\tens{Q}}
\newcommand{\tQt}{\tQ^{\T}}
\newcommand{\thQ}{\widehat{\tQ}}
\newcommand{\thQt}{\thQ^{\T}}
\newcommand{\hsigma}{\ensuremath{\hat{\sigma}}}
\newcommand{\tZ}{\tens{Z}}
\newcommand{\tZt}{\tZ^{\T}}
\newcommand{\thZ}{\widehat{\tZ}}
\newcommand{\thZt}{\thZ^{\T}}

\newcommand{\tW}{\tens{W}}
\newcommand{\tWt}{\tW^{\T}}
\newcommand{\thW}{\widehat{\tW}}
\newcommand{\thWt}{\thW^{\T}}

\newcommand{\tP}{\tens{P}}
\newcommand{\tPt}{\tP^{\T}}
\newcommand{\thP}{\widehat{\tP}}
\newcommand{\thPt}{\thP^{\T}}

\newcommand{\upbb}[2]{\ensuremath{{#1}^{(#2)}}}

\newcommand{\tWh}{\upbb{\tW}{h}}
\newcommand{\tWht}{( \tWh )^{\T}}
\newcommand{\thWh}{\widehat{\upbb{\tW}{h}}}
\newcommand{\thWht}{(\thWh)^{\T}}

\newcommand{\tCh}{\upbb{\tC}{h}}
\newcommand{\tCht}{(\tCh)^{\T}}
\newcommand{\thCh}{\widehat{\upbb{\tC}{h}}}
\newcommand{\thCht}{(\thCh)^{\T}}

\newcommand{\tWl}{\upbb{\tW, \ell}}
\newcommand{\tWlt}{( \tWh )^{\T}}
\newcommand{\thWl}{\widehat{\upbb{\tW}{\ell}}}
\newcommand{\thWlt}{(\thWh)^{\T}}

\newcommand{\tCl}{\upbb{\tC}{\ell}}
\newcommand{\tClt}{(\tCh)^{\T}}
\newcommand{\thCl}{\widehat{\upbb{\tC}{\ell}}}
\newcommand{\thClt}{(\thCh)^{\T}}

\newcommand{\qb}{Q_{\tB}}
\newcommand{\qv}{Q_{\tV}}
\newcommand{\pb}{P_{\tB}}
\newcommand{\tpb}{\tilde{P}_{\tB}}
\newcommand{\pv}{P_{\tV}}
\newcommand{\rbb}{R_{\matb}}
\newcommand{\rvv}{R_{\matr}}

\newcommand{\mat}[1]{\ensuremath{\mathbf{#1}}}
\newcommand{\matA}{\mat{A}}
\newcommand{\matX}{\mat{X}}
\newcommand{\matXt}{\matX^{\T}}
\newcommand{\matAt}{\matA^{\T}}
\newcommand{\bmatA}{\bar{\matA}}
\newcommand{\bmatAt}{\bmatA^{\T}}
\newcommand{\mata}{\mat{a}}
\newcommand{\matV}{\mat{V}}
\newcommand{\matVt}{\ensuremath{\matV^{\T}}}
\newcommand{\matU}{\mat{U}}
\newcommand{\matS}{\mat{S}}
\newcommand{\matSt}{\ensuremath{\matS^{\T}}}
\newcommand{\mats}{\mat{s}}
\newcommand{\hmats}{\hat{\mats}}
\newcommand{\matr}{\mat{r}}
\newcommand{\maty}{\mat{y}}
\newcommand{\amatr}{\vec{\matr}}

\newcommand{\matZ}{\mat{Z}}
\newcommand{\matZt}{\ensuremath{\matZ^{\T}}}
\newcommand{\matY}{\mat{Y}}
\newcommand{\matYt}{\ensuremath{\matY^{\T}}}

\newcommand{\matM}{\mathbf{M}}
\newcommand{\matMt}{\matM^{\T}}
\newcommand{\matW}{\mathbf{W}}
\newcommand{\matWt}{\ensuremath{\matW^{\T}}}
\newcommand{\T}{\mat{T}}

\renewcommand{\u}{\mat{u}}
\newcommand{\ut}{\ensuremath{\u^{\T}}}
\newcommand{\hu}{\ensuremath{\hat{\u}}}
\renewcommand{\v}{\mat{v}}
\newcommand{\hv}{\ensuremath{\hat{\v}}}
\newcommand{\rrho}{\ensuremath{\bs{\rho}}}
\newcommand{\pphi}{\ensuremath{\bs{\varphi}}}
\newcommand{\w}{\mat{w}}
\newcommand{\wt}{\ensuremath{\w^{\T}}}
\newcommand{\x}{\mat{x}}
\newcommand{\p}{\mat{p}}
\newcommand{\y}{\mat{y}}
\newcommand{\bmaty}{\bar{\y}}
\newcommand{\bmatx}{\bar{\x}}
\newcommand{\z}{\mat{z}}
\newcommand{\matB}{\mat{B}}
\newcommand{\matI}{\mat{I}}
\newcommand{\matBt}{\matB^{\T}}

\newcommand{\matb}{\mat{b}}
\newcommand{\matbt}{\matb^{\T}}
\newcommand{\mate}{\mat{e}}
\newcommand{\matet}{\mate^{\T}}

\newcommand{\bsc}{\bs{c}}
\newcommand{\bsct}{\bsc^{\T}}

\newcommand{\tlX}{\tilde{X}}
\newcommand{\tlY}{\tilde{Y}}

\newcommand{\RR}{\mathbb{R}}
\newcommand{\CC}{\mathbb{C}}
\newcommand{\FF}{\mathbb{F}}

\newcommand{\dotp}[1]{\langle #1 \rangle}
\newcommand{\dotps}[1]{\dotp{#1}^2}
\newcommand{\FDot}[1]{\dotp{#1}_{F}}
\newcommand{\FDotS}[1]{\dotps{#1}_{F}}
\newcommand{\TNorm}[1]{\|#1\|_{2}}
\newcommand{\FNorm}[1]{\|#1\|_{F}}
\newcommand{\NNorm}[1]{\|#1\|_{*}}
\newcommand{\FNormS}[1]{\FNorm{#1}^2}
\newcommand{\TNormS}[1]{\TNorm{#1}^2}
\newcommand{\trace}[1]{\ensuremath{\operatorname{Tr}(#1)}}
\newcommand{\ftr}[1]{\ensuremath{\operatorname{f-Tr}(#1)}}

\newcommand{\clr}{\noun{clr}}
\newcommand{\rclr}{\noun{rclr}}
\newcommand{\tsvdm}{\textsc{tsvdm}\xspace}
\newcommand{\tcam}{\textsc{tcam}\xspace}
\newcommand{\tca}{\textsc{tca}\xspace}
\newcommand{\tsvdmii}{\ensuremath{t}-\textsc{svdmii}\xspace}

\newcommand{\tsub}[1]{\ensuremath{\times_{#1}}}
\newcommand{\tsM}{\ensuremath{\tsub{3}\matM}}
\newcommand{\tsMinv}{\ensuremath{\tsub{3}\matM^{-1}}}
\newcommand{\tprod}[1]{\star_{{\scriptscriptstyle{#1}}}}
\newcommand{\Mprod}{\tprod{\matM}}
\newcommand{\mm}{\tprod{\matM}}
\newcommand{\muni}{\ensuremath{\Mprod{\textnormal{-unitary}}}\xspace}
\newcommand{\morth}{\ensuremath{\Mprod{\textnormal{-orthogonal}}}\xspace}
\newcommand{\pmorth}{\ensuremath{\textnormal{pseudo }\Mprod{\textnormal{-orthogonal}}}\xspace}
\newcommand{\Pmorth}{\ensuremath{\textnormal{Pseudo }\Mprod{\textnormal{-orthogonal}}}\xspace}
\newcommand{\rnk}{\ensuremath{\operatorname{rank}}}
\newcommand{\Npr}{\ensuremath{N^{\prime}}}

\newcommand{\Mpinv}{\mathbf{+}}
\newcommand{\mmpinv}{\ensuremath{\Mprod}\textnormal{-pseudo inverse}\xspace}

\algnewcommand{\IfThenElse}[3]{
	\State \algorithmicif\ #1\ \algorithmicthen\ #2\ \algorithmicelse\ #3}
\algnewcommand\Input{\item[\textbf{Input:}]}%
\algnewcommand\algorithmicinput{\textbf{Input:}}
\algnewcommand\INPUT{\item[\algorithmicinput]}
\algnewcommand\SState{\State \hskip-1.em }
\algnewcommand\algorithmicswitch{\textbf{switch}}
\algnewcommand\algorithmiccase{\textbf{case}}
\algnewcommand\algorithmicassert{\texttt{assert}}
\algnewcommand\Assert[1]{\State \algorithmicassert(#1)}

\algnewcommand\Output{\item[\textbf{Output:}]}%

\newcommand{\samplemode}{\emph{sample mode} }
\newcommand{\featuremode}{\emph{feature mode} }
\newcommand{\omx}{{\it omics} }
\newcommand{\LFB}{LFB }

\newcommandx{\unsure}[2][1=]{\todo[linecolor=red,backgroundcolor=red!25,bordercolor=red,#1]{#2}}
\newcommandx{\change}[2][1=]{\todo[linecolor=blue,backgroundcolor=blue!25,bordercolor=blue,#1]{#2}}
\newcommandx{\info}[2][1=]{\todo[linecolor=OliveGreen,backgroundcolor=OliveGreen!25,bordercolor=OliveGreen,#1]{#2}}
\newcommandx{\improvement}[2][1=]{\todo[linecolor=Plum,backgroundcolor=Plum!25,bordercolor=Plum,#1]{#2}}
\newcommandx{\thiswillnotshow}[2][1=]{\todo[disable,#1]{#2}}

\newcommand{\COMMENTFILE}{child_docs/commentfile.tex}
\newcommand{\ecomment}[2]{\unsure{{\bf (line \number\inputlineno,#1,#2)}:\newline%
									\ExecuteMetaData[\COMMENTFILE]{#1}}\xspace}

\newcommand{\pbxstd}{Fig.1b\xspace}
\newcommand{\pbxtca}{Fig.1c\xspace}
\newcommand{\pbxfunnel}{Fig.1d\xspace}
\newcommand{\pbxbars}{Fig.1e\xspace}

\newcommand{\fibersctf}{Fig.fb\xspace}
\newcommand{\fiberstca}{Fig.1g\xspace}
\newcommand{\fibersfunnel}{Fig.1h\xspace}
\newcommand{\fiberstimeseries}{Fig.1i\xspace}
\newcommand{\fibersheatmap}{Fig.1j\xspace}

\newcommand{\ibdrocs}{Fig.2a\xspace}
\newcommand{\ibdloadings}{Fig.2b\xspace}
\newcommand{\ibdtca}{Fig.2c\xspace}
\newcommand{\ibdtimeseries}{Fig.2d\xspace}

\newcommand{\snydertca}{Fig.2e\xspace}
\newcommand{\snyderheatmap}{Fig.2f\xspace}

%% file: child_docs/clean_text.tex
\section*{Introduction}
Precision health and medicine aim to provide disease treatment, pre-clinical detection and prevention, while taking into consideration the individual genetic variability, environment and lifestyle. 
Recent developments in high-throughput methodologies enable the assessment of molecular entities from biological samples on a global scale at steadily decreasing costs, allowing to conduct biological and clinical studies at previously unfeasible magnitude, including the number of biological repetitions and molecules quantified ~\cite{Adlung2021}. 
A consequence of the increased availability of omics methods, is the possibility to conduct large-scale longitudinal studies prospectively following up the participants. 
In particular, longitudinal omics profiling, combined with clinical measurements, enable us to detect and understand individual changes from baseline, improving personalized health and medicine by using tailored therapies ~\cite{SchsslerFiorenzaRose2019}.

\vspace{.25cm}
\noindent
Yet, despite the surge of longitudinal multi-omics studies, the tool-set for such analyses remains limited to date, with only a handful of applicable software suitable for specific tasks ~\cite{Metwally2018, Shields-Cutler2018, Plantinga2017}. 
Recently, an impressive advancement in the use of tensor factorization methods for time series analysis emerged, allowing trajectory analysis for microbiome data~\cite{Martino2020,Delannoy-Bruno2021} as well as neural dynamics~\cite{Williams2018}. 
Generally referred to as tensor component analysis (TCA)~\cite{Williams2018}, these multiway dimensionality reduction methods for omics data are based on CANDECOMP/PARAFAC (CP) factorization~\cite{Hitchcock1927,Harshman1970}, which dramatically limits the ability to apply machine learning (ML) algorithms, as it does not allow for straightforward mapping of unseen data points to the reduced space.
In addition, CP-based TCA requires choosing the number of components (dimensions) to be considered, since different choices may result in significantly different transformations of the data, additional uncertainties in analyzing complex information are introduced.

\noindent
Here we present \tcam{}, a new method for dimensionality reduction which provides answers the unmet need of trajectory analysis of longitudinal omics data. 
Our novel method is based on a cutting-edge mathematical framework (the M-product between tensor), which allows for a natural generalization of the notion of singular value decomposition (SVD) for matrices ($2^{nd}$ order tensors) to higher order tensors~\cite{Kilmer}. 
We show that \tcam{} outperforms traditional methods, as well as  recent - microbiome specific - tensor factorization methods for longitudinal microbiome data analysis, both in identifying distinct trajectories between different phenotypic groups, and in highlighting significant temporal variation in bacterial entities.
Furthermore, we demonstrate the versatility of \tcam{} by applying it to a proteomics dataset, uncovering new insights that were not disclosed in the original paper, showing that our method works not only for microbiome datasets but is also applicable for a wide array of longitudinal omics data. 
Finally, we show that in contrast to CP-based TCA methods, our methodology can also be applied for straightforward ML tasks on omics data.

\input{child_docs/figure_panels/cartoon_fig}

\section*{Results}
For the purpose of this study we utilized four different longitudinal datasets~\cite{Suez2018,Schirmer2018,Sailani2020,Deehan2020}, which include 16S rDNA microbiome analysis, shotgun metagenomics and proteomics data. The first study~\cite{Suez2018} investigated the reconstitution of the gut microbiome in healthy individuals following antibiotic administration, by comparing a 21 day-long probiotics supplementation (PBX), autologous fecal microbiome transplantation (aFMT) derived from a pre-antibiotics treated sample, or spontaneous recovery (CTR), using longitudinal sampling from baseline until reconstitution (n=17). 
The second study~\cite{Deehan2020} is an interventional experiment, testing the impact of different resistant starch4 (RS4) structures on microbiome composition, in which stool samples were collected each week during a five weeks long trial (n=40). The four arms of this experiment were defined by the source of fibers: tapioca and maize groups represent sources of fermentable fibers, while potato and corn groups mostly contain fibers that are inaccessible for microbiome degradation thus, considered control groups. 
The third study~\cite{Schirmer2018} evaluated three different treatments for pediatric ulcerative colitis (UC) during a one-year follow-up, in which the stool microbiome was examined in four specific time-points (n=87). 
Finally, the fourth study~\cite{Sailani2020} constituted a unique longitudinal experimental setting of 105 healthy volunteers to explore the influence of seasons in biological processes. 
\noindent
To this aim, the authors performed immune profiling, proteomics, metabolomics, transcriptomics and metagenomics, collecting approximately 12 samples from each participant during the time course of four years.
\vspace{.25cm}

\phantomsection \label{parpbx}
\noindent
First, we sought to demonstrate the advantages of \tcam{} over the traditional methods for microbiome analysis by applying \tcam{} on the data from the antibiotic reconstitution project~\cite{Suez2018}. 
According to the original study, participants were split into three study arms (PBX, aFMT and CTR) and stool samples were collected at baseline (days 0 to 6), antibiotics treatment (days 7 to 13) and the intervention phase (days 14 to 42). 
We used Principal Component Analysis (PCA) as a comparison reference for our \tcam{} method, as it constitutes a traditional gold-standard for microbiome analysis. 
Indeed, when we applied PCA to all of the time points, it resulted in a reduced representation that was highly affected by inter-individual differences~(\Cref{fig:postabx.pca.all}), while temporal intra-individual information in longitudinal data analysis is masked the inter-individual variability. The truncated representation following PCA brings little addition to the information obtained from baseline samples, as the distances between the samples across all samples are tightly correlated with the baseline distances between samples. ~(\Cref{fig:postabx.distance.reg,sfig:postabx.pca.baseline}). 
Similarly, the per-phase perspective of the data did not capture the trend of composition changes, but a mere snapshot of temporal trends~(\Cref{sfig:postabx.pca.group.BAS,sfig:postabx.pca.group.INT,sfig:postabx.pca.group.ABX}).

\noindent
In contrast, analysis by \tcam{} generated a temporally coherent representation of the data ~(\Cref{fig:postabx.tcam.factors}), with only one single point representing the full trajectory of a subject throughout the entire experiment.
Additionally, the \tcam{} factor scores approximate the true distances between trajectories ~(\cref{app:discussion}) providing an easy interpretation, which is amenable for multivariate hypothesis testing methods. 
We performed a PERMANOVA test to reveal significant differences between trajectories in the FMT and PBX groups (\Cref{fig:postabx.tcam.factors}; p<0.05), in agreement with the original findings in the study.
In addition, we highlighted the bacterial features contributing for this distinct separation between the groups, which the original study could not detect~(\Cref{fig:postabx.tcam.loadings}), with five of the most contributing bacteria for the significant separation between the groups were the actual probiotic species consumed by the PBX group~\cite{Suez2018}.
Furthermore, we harnessed the power of \tcam{} as part of a pruning strategy, by considering only top \tcam{} loadings for a univariate linear mixed effect model (lmer, methods). Using our pruning strategy, we managed to discover twenty-three new bacterial features that significantly differ between the groups (two of these are probiotic species, ~\Cref{fig:postabx.timeser.tcam}), with an overlap of eight species discovered with and without pruning (\Cref{fig:postabx.timeser.mutual}). 
The pruning strategy failed in detection of three species, that otherwise were discovered (\Cref{fig:postabx.timeser.all,fig:postabx.bars}).

\input{child_docs/figure_panels/postabx_figure}


\phantomsection \label{par:gemelli}
\noindent
Next, we evaluated the \tcam's performance against the new state-of-the-art tensor factorization method Gemelli~\cite{Martino2020}. Unlike \tcam{}, Gemelli is designed specifically for 16S amplicon sequencing data, as it utilizes mathematical properties that are unique to such data. 
For this reason, we used the RS4 interventional dataset, comparing the effect of four different types of RS4 fiber administration, tapioca, maize, corn and potato, on the microbiome composition~\cite{Deehan2020}. 
In the original paper, the authors noticed significant changes in specific time-points in the tapioca and maize groups, but no apparent trends of changes in the microbiome composition were reported. 

\noindent
Initial analysis with Gemelli, failed to identify any significant differences between trajectories of the groups (PERMANOVA; p>0.05~\Cref{fig:fibers.gemelli.factors}), however, using \tcam{}, we were able to detect a significant trajectory for the maize group compared to all other groups, but not for the tapioca group (PERMANOVA; p<0.05,~\Cref{fig:fibers.tcam}). 
We then applied our pruning strategy, as previously described, and managed to identify four distinct bacteria featuring a statistically significant trend throughout time, not detected in the traditional methods or by Gemelli (\Cref{fig:fibers.gemelli.funnel,fig:fibers.timeseries}). 
Moreover, using the top loadings of \(F_3\) (see methods), we highlighted additional features that did not pass the FDR-correction significance threshold in the previous strategy, demonstrating patterns of increasing bacteria in the form of {\it Lachnospiraceae} (p<0.05) in the maize group and {\it P. distasonis} (p<0.05)  in the tapioca group (\Cref{fig:fibers.heatmap}).

\noindent
Taken all together, we can determine that \tcam{} outperformed both traditional and the up-to-date longitudinal analysis workflow for time-series longitudinal trajectory analysis, identifying signals in temporal multivariate complex data that the aforementioned methods could not. Additionally, \tcam{} performed exceptionally well as a feature pruning strategy, which reduces the features of interest and enables a sensitive detection of significant features.

\input{child_docs/figure_panels/fibers_figure}

\vspace{.25cm}
\phantomsection \label{par:snyder}
\noindent
We turned to assess \tcam's applicability to general omics data, distinct from metagenomics, by applying our method to the proteomics data set~\cite{Sailani2020}. 
In this study, twelve samples were collected from 105 healthy participants during three-years follow-up (collection of one sample every three months), and tested the seasonal trends of the microbiome, transcriptome, metabolome and proteome. 
We utilized a subset of this study’s data, focusing only on proteomics of individuals featuring information related to their insulin sensitivity (IS) or insulin resistance (IR). 
Instead of addressing seasonal patterns, which were the main focus of the original work, we addressed the differences between proteome trajectories of the IR and IS groups throughout the three-years follow-up. 
Using \tcam{}, we detected a significant separation between the two groups based on \(F_1\) (t-test; p<0.05), indicating a different trajectory of these two population across time (\Cref{fig:iris.tcam}).
Similar to our analysis framework above, we turned to the top loadings contributing to this signature (\Cref{fig:iris.heatmap}). 
Among the top ranked proteins, we could easily notice angiotensinogen (AGT), which levels are associated with IS~\cite{Underwood2011}, paraoxonase-1 (PON1) which was found to reduce IR in mice~\cite{KorenGluzer2013}, apolipoprotein-3 (APOC3), highly associated with IR~\cite{Chung2015}, highlighting different trajectories for these proteins among the two groups and increasing levels of AHSG in the IS which indeed tightly associated with IS (DOI: 10.2337/diacare.29.04.06.dc05-1938) \input{child_docs/figure_panels/iris_figure}

\vspace{.25cm}
\phantomsection \label{par:ibd}

\noindent
Finally, we demonstrated the predictive power unleashed by \tcam from longitudinal sampling in the context of supervised ML tasks such as classification.
To this aim, we utilized a 16S rDNA microbiome dataset from a study by Schirmer et. al., which includes stool samples derived from pediatric UC patients monitored for 52 weeks under three different treatments, and characterized for microbial dynamics along disease course in light of host response to each of the applied treatments~\cite{Schirmer2018}. 
When considering the complete time course, as well as, single timepoint snapshots, we found no clear separation based on the remission status of the UC patients (PERMANOVA, p>0.05;~\Cref{fig:ibd.remission.projections.tcam,fig:ibd.remission.projections.w12,fig:ibd.remission.projections.w52}), making the task of predicting the remission status using temporal microbiome data highly challenging.

\noindent
To overcome this challenge, we used \tcam{} as a feature engineering tool in the generation of a machine learning classifier. 
In short, we executed five cross-validation Multi Layer Perceptron (MLP) classifier  using the \tcam{} transformed training set data, the same transformation was used in an out-of-sample extension manner to transform the validation set. 
The MLP had shown an impressive mean Area Under The Receiver Operating Characteristic (AUROC) of (AUROC = 0.78, ~\Cref{fig:ibd.tcam.roc} )   for the classification of remission status. 
Moreover, we were able to preserve the original feature importance contribution prior to \tcam{}, thus highlighting key bacteria that differs between remission status; 
Using standard variable-importance API, we were able to pin-point specific taxa like {\it Anaerococcus} and {\it B.fragilis}  whose trajectories data made the highest contribution to the decision making process (\Cref{fig:ibd.tcam.important}),  
that could not have been identified using a per time-point comparison as done in the original paper. 
To validate the \tcam{}-based feature importance results, we then applied \tcam{} on the same dataset, while taking into account only the top \%5 ranked features  discovered for the MLP classifier pipeline. 
From this reduced view, we were able to identify clustering according to remission states (\Cref{fig:ibd.pruned.tcam}), validating the power of \tcam{} classification and the preservation of feature importance scores(\Cref{fig:abundances.important.features}).
\input{child_docs/figure_panels/ibd_figure}

\noindent
We next compared the performances of the MLP \tcam{} pipeline with those of an MLP pipeline with PCA on single time points of week 12 and week 52 (see \cref{sec:online.methods}), changing only the feature engineering step while using a fixed MLP architecture. Indeed, as we expected, \tcam{} outperformed PCA as a feature engineering technique, with both weeks 12 and 52, displaying a lower mean AUROC than \tcam{} based workflow (W12=0.61,W52=0.68, \Cref{fig:ibd.roc.w12,fig:ibd.roc.w52}).



\phantomsection \label{short summary}
\vspace{.25cm}
\noindent
Overall, in this section, we demonstrated the utility of \tcam{} across multi-omics data and as an integral part of ML pipelines. In both scenarios we were able to show that the application of \tcam{} had a unique added value over the traditional methods for longitudinal analysis. Moreover, we displayed the ability of \tcam{} to exploit the full power of ML tools in longitudinal datasets, as feature engineering tool, with the preservation of the original features contributions.

\section*{Discussion}
In this work, we presented \tcam{}, a novel dimensionality reduction method for longitudinal 'omics data analysis, constructed on top of solid tensor-tensor algebra innovations. We demonstrated that \tcam outperforms traditional and state-of-the-art methods for longitudinal analysis dimensionality reduction, both in terms of signature detection and by pruning for meaningful features. In addition, we showed that \tcam{} is applicable to diverse omics types, including amplicon and shotgun sequencing as well as proteomics. Furthermore, unlike other tensor factorization methods, \tcam entertains a natural out-of-sample extension formula, making it suitable for prediction tasks in complex experimental designs as a drop-in feature engineering utility within ML workflows. We have showed that we can preserve the feature importance contribution of the original features, even when \tcam{} is applied.

\noindent
To our knowledge \tcam{} is the first tensor component analysis framework that is guarantied, within the specific choice of domain transformation, to maximize the variance of the latent representation while keeping the distortion minimal. Thus, \tcam{} is amenable for traditional downstream applications often used in biological data analysis, such as multivariate hypothesis testing and ML workflows.  

\noindent     
While \tcam{} proves to be an extremely useful tool for all longitudinal analysis experimental designs, it relies on a fully sampled cohorts, where all participants provide a comparable number of samples and at similar time points in order to extract insightful biological signatures. Prior to usage of \tcam{}, a user will have to complete the missing time points by any method he chooses, as we have done throughout this study (see methods).

\noindent
Looking forward, the mathematical properties of \tcam{} are supposed to enable us not only to perform a trajectory analysis across time, but also in a spatial manner across a geographical landscape. Moreover, it is possible to employ a \tcam{} decomposition on higher order tensors, allowing for better understanding of even more complex experimental designs, such as incorporation of space and time together.

\noindent
Overall, this novel approach answers the important unmet need of longitudinal 'omics data analysis tool-kits that enable trajectory analysis, and is available \footnote{\href{https://github.com/UriaMorP/mprod_package}{https://github.com/UriaMorP/mprod\_package}} to the wide community as a simple, one-stop-shop Python implementation, that is compatible with the highly popular scikit-learn package.
We hope that the application of \tcam{} could help derive deep insights from large-scale, longitudinal and multi-omics data thus promoting personalized medicine, leading to the development of tailored treatments and preventive strategies for human diseases.

\vspace{\fill}
\section*{Acknowledgments}
E.E. is supported by the Leona M. and Harry B. Helmsley Charitable Trust, Adelis Foundation, Pearl Welinsky Merlo Scientific Progress Research Fund, Park Avenue Charitable Fund, Hanna and Dr. Ludwik Wallach Cancer Research Fund, Daniel Morris Trust, Wolfson Family Charitable Trust and Wolfson Foundation, Ben B. and Joyce E. Eisenberg Foundation, White Rose International Foundation, Estate of Malka Moskowitz, Estate of Myron H. Ackerman, Estate of Bernard Bishin for the WIS-Clalit Program, Else Kröener-Fresenius Foundation, Jeanne and Joseph Nissim Center for Life Sciences Research, A. Moussaieff, M. de Botton, Vainboim family, A. Davidoff, the V. R. Schwartz Research Fellow Chair and by grants funded by the European Research Council, Israel Science Foundation, Israel Ministry of Science and Technology, Israel Ministry of Health, Helmholtz Foundation, Garvan Institute of Medical Research, European Crohn’s and Colitis Organization, Deutsch-Israelische Projektkooperation, IDSA Foundation and Wellcome Trust. E.E. is the incumbent of the Sir Marc and Lady Tania Feldmann Professorial Chair, a senior fellow of the Canadian Institute of Advanced Research and an international scholar of the Bill \& Melinda Gates Foundation and Howard Hughes Medical Institute. H.A. is supported by the Israel Science Foundation, US-Israel Binational Science Foundation, and IBM Faculty Award. This research was partially supported by the Israeli Council for Higher Education (CHE) via the Weizmann Data Science Research Center.
\vfill

\section*{Author contribution}
{\bf U.M.} 
conceived the study, established theoretical results,
wrote the software package, 
analyzed the data, 
generated the figures 
and wrote the manuscript; 
{\bf Y.C. and R.V.-M.} 
analyzed the data, 
gave biologically and clinically meaningful interpretation for the results, 
generated the figures 
and wrote the manuscript;
{\bf D.K.} assisted with data analysis and biological interpretation, 
wrote the manuscript;
{\bf E.E.} conceived the study, 
mentored the participants and wrote the manuscript;
{\bf H.A.} conceived the study, 
supervised the theoretical aspects, 
mentored the participants and wrote the manuscript.
\vfill

\section*{Competing interest statement}
E.E. is a scientific founder of DayTwo and BiomX, and a payed consultant to Roots GmbH. H.A. is an inventor of U.S. patent US10771088B1 which discloses the {\sc TSVDM} decomposition which is the basis for \tcam{}. U.S. patent US10771088B1 is assigned to Tel Aviv Yafo University,  International Business Machines Corp and Tufts University. Inventors are Lior Horesh, Misha E. Kilmer, H.A., and Elizabeth Newman. The remaining authors declare no competing interests. U.M., Y.C., R.V.-M. and D.K. do not have any financial or non-financial competing interest.
\vfill

\section*{Code and data availability}
The code for the analysis presented in the paper can be found in the Github repository~\footnote{\href{https://github.com/UriaMorP/tcam_analysis_notebooks}{https://github.com/UriaMorP/tcam\_analysis\_notebooks}}. No new data was generated in this study. 


%% file: child_docs/figure_panels/cartoon_fig.tex
\begin{figure}[H]
    \hspace*{-.25in}
    \centering
    \includegraphics[scale=.9]{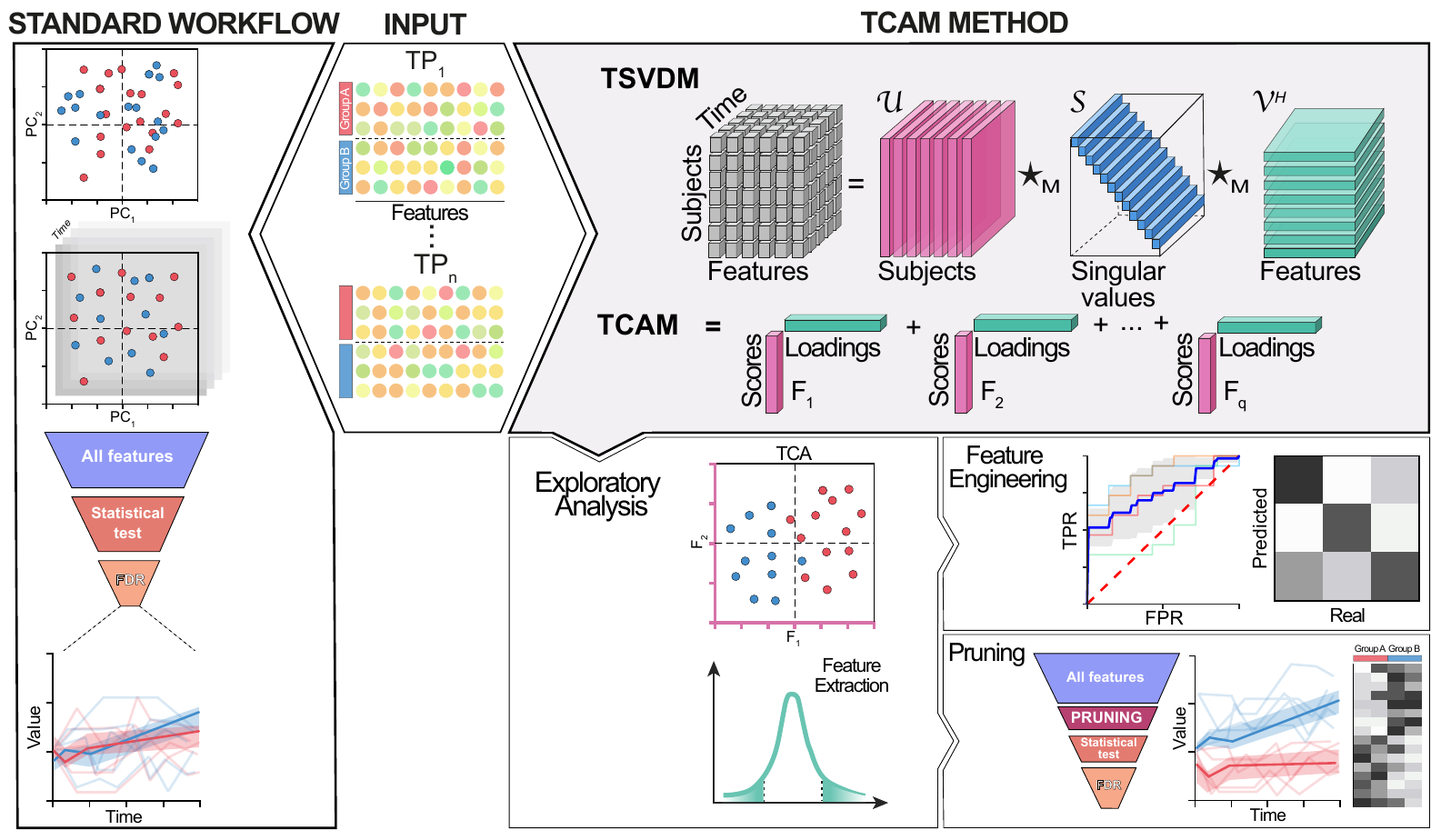}
    \caption{{\bf Illustration of settings and workflow.}  Center rhombus describes the typical data produced in a longitudinal experiment, where $p$ dimensional samples are collected from $m$ subjects across $n$ timepoints.  
    {\bf Fig.1a} shows our new method for obtaining tensor component analysis factors from an explicit rank truncated \tsvdm of the data.   
    {\bf Figures.1b-f} illustration of some prominent outcomes of \tcam application, including 
    {\bf Fig.1b} ordination plot for temporal trajectories where each point approximate the complete $p$-dimensional $n$-timepoints long time series course of each subject in the study.
    {\bf Fig.1c} Illustration of factor loadings distribution (top) and a summarised view for the highest magnitude contributors (bottom). The right-pointing arrow represents the applicability of factor loadings information to feature pruning, prior to univariate time-series analysis
    {\bf Fig.1d} (Left)  Pruning funnel, describing the \tcam based feature selection, from the initial number of features, through the selection of features corresponding to highest magnitude loadings of factors of interest, followed by application of hypothesis testing and false discovery rate correction. (Right) Illustration of an expected outcome of the strategy - discovery of feature exhibiting significantly different temporal trends between two (or more) experimental groups.
    {\bf Fig.1e,f} Illustration of ROC curve {\bf e} and confusion matrix {\bf f}, exemplifying  \tcam as a feature engineering step that seamlessly integrates in any standard ML workflow.
    {\bf Fig.1g,h} Illustrates the traditional, matrix-based longitudinal 'omics data analysis workflow. 
    {\bf Fig.1g} Illustration of a collection of ordination scatterplots (such as PCA), corresponding to one possible marginalization of the data.
    {\bf Fig.1h} Demonstration of features funnel when no pruning strategy is applied (top), and expected outcome of univariate time-series analysis when applied to all features (bottom).
    }
\end{figure}

%% file: child_docs/figure_panels/postabx_figure.tex
\begin{figure}[H]
    \hspace*{-.25in}
    \centering
    \begin{subfigure}[t]{0.31\textwidth}
        \centering
        \caption{\label[cap]{fig:postabx.pca.all}}
        \includegraphics[scale=.8]{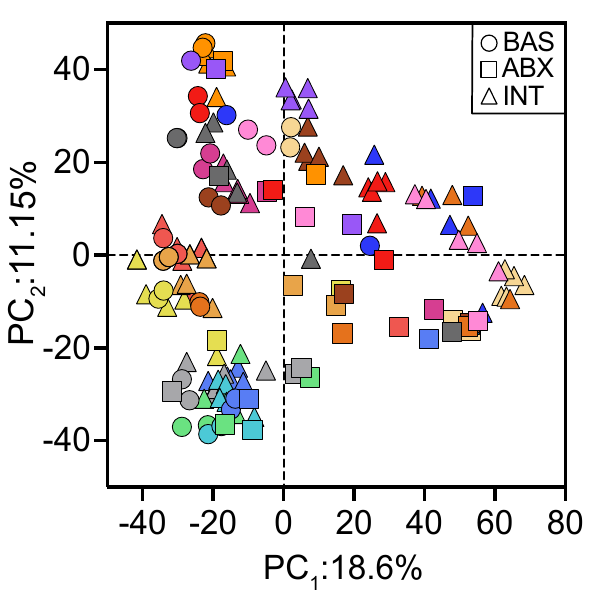}
    \end{subfigure}\hfill%
    \begin{subfigure}[t]{0.31\textwidth}
        \centering
        \caption{\label[cap]{fig:postabx.distance.reg}}
        \includegraphics[scale=.8]{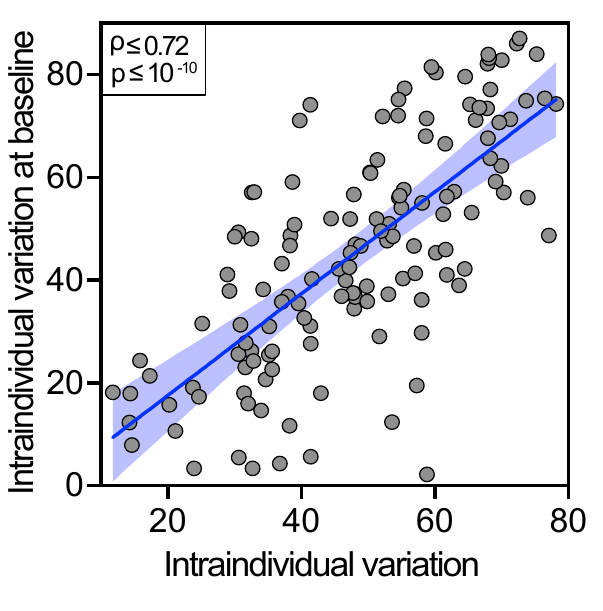}
    \end{subfigure}\hfill%
    \begin{subfigure}[t]{0.31\textwidth}
        \centering
        \caption{\label[cap]{fig:postabx.tcam.factors}}
        \includegraphics[scale=.8]{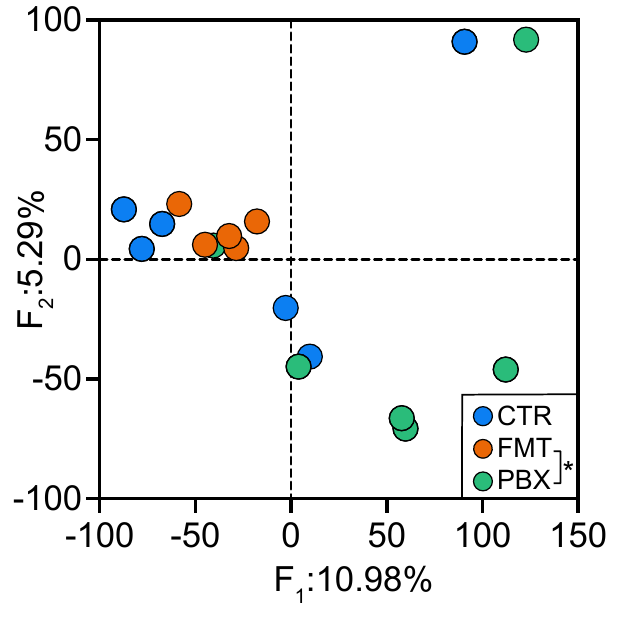}
    \end{subfigure}
    \begin{subfigure}[t]{0.24\textwidth}
        \centering
        \caption{\label[cap]{fig:postabx.tcam.loadings}}
        \includegraphics[scale=.8]{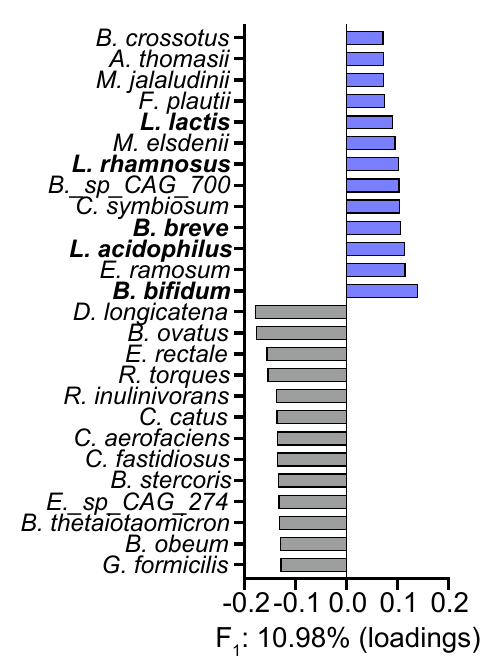}
    \end{subfigure}\hfill%
    \hspace*{-.25in}
    \begin{subfigure}[t]{0.24\textwidth}
        \centering
        \caption{\label[cap]{fig:postabx.funnel}}
        \includegraphics[scale=.7]{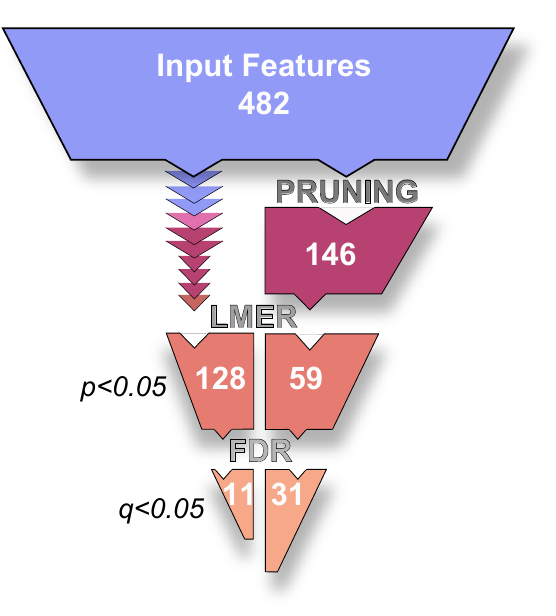}
    \end{subfigure}%
    \begin{subfigure}[t]{0.5\textwidth}
        \centering
        \caption{\label[cap]{fig:postabx.bars}}
        \includegraphics[scale=.9]{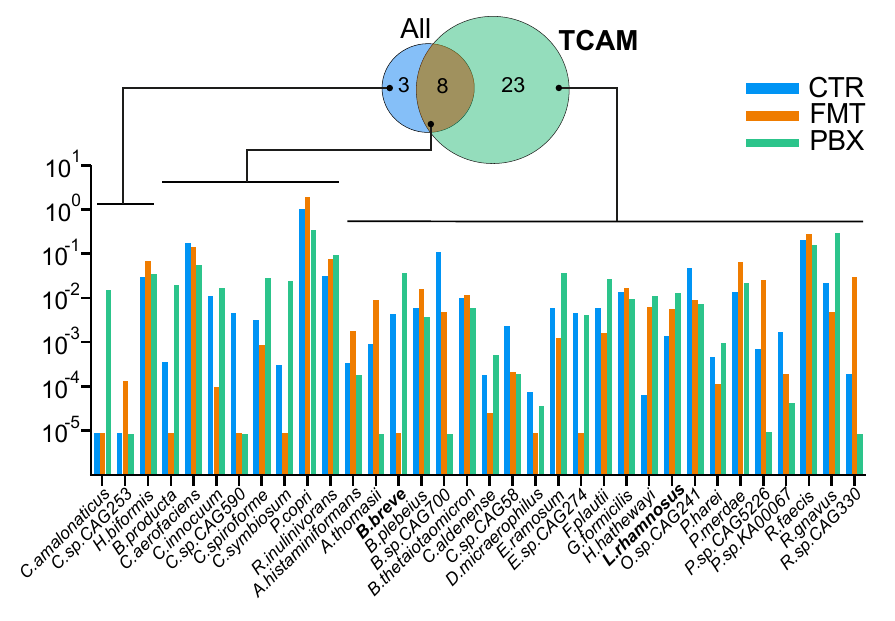}
    \end{subfigure}
    
    \caption{{\bf Comparison of \tcam with existing matrix based methods.} 
    \Cref{fig:postabx.pca.all} PCA plot of all timepoints, colored according to participant. 
    \Cref{fig:postabx.distance.reg} Regression line of mean distance between subjects at all timepoints (x) and at baseline (y). Distances computed using PC$_1$ and PC$_2$. 
    \Cref{fig:postabx.tcam.factors} Scatterplot of the first \tcam factors. Each point represents the whole trajectory of a participant. 
    \Cref{fig:postabx.tcam.loadings} Bar graph showing top 2.5\% features contributing to \(F_1\)s variation. 
    \Cref{fig:postabx.funnel} Comparison of discovery rates for univariate hypothesis testing (lmer), between naive (left) and  \tcam based pruning (right) workflow. 
    \Cref{fig:postabx.bars} Venn diagram and bar graphs. Bar show are per-subject AUC for all detected bacteria (q<0.05). 
    Venn diagram relates each bacterium to the workflow it was detected in. 
    Bars represent medians.}
\end{figure}


%% file: child_docs/figure_panels/fibers_figure.tex

\begin{figure}[H]
    \hspace*{-.25in}
    \centering
    \begin{subfigure}[t]{0.31\textwidth}
        \centering
        \caption{\label[cap]{fig:fibers.gemelli.factors}}
        \includegraphics[scale=.8]{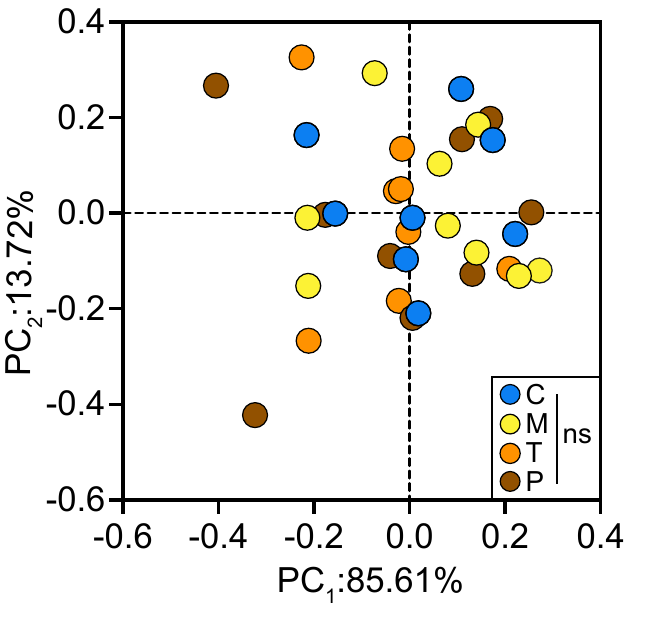}
    \end{subfigure}\hfill%
    \begin{subfigure}[t]{0.31\textwidth}
        \centering
        \caption{\label[cap]{fig:fibers.tcam}}
        \includegraphics[scale=.8]{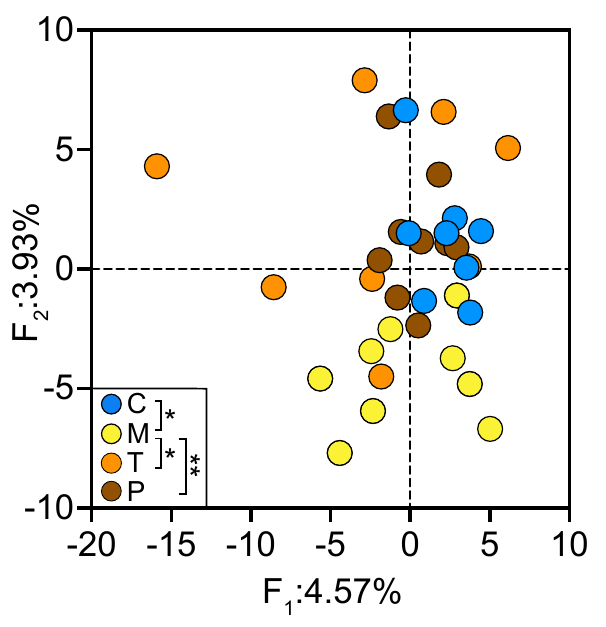}
    \end{subfigure}\hfill%
    \begin{subfigure}[t]{0.31\textwidth}
        \centering
        \caption{\label[cap]{fig:fibers.gemelli.funnel}}
        \includegraphics[scale=.8]{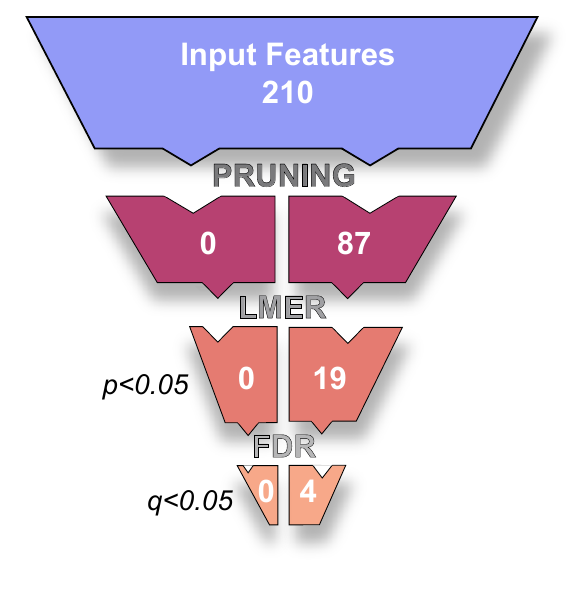}
    \end{subfigure}
    \begin{subfigure}[t]{0.62\textwidth}
        \centering
        \caption{\label[cap]{fig:fibers.heatmap}}
        \includegraphics[scale=.8]{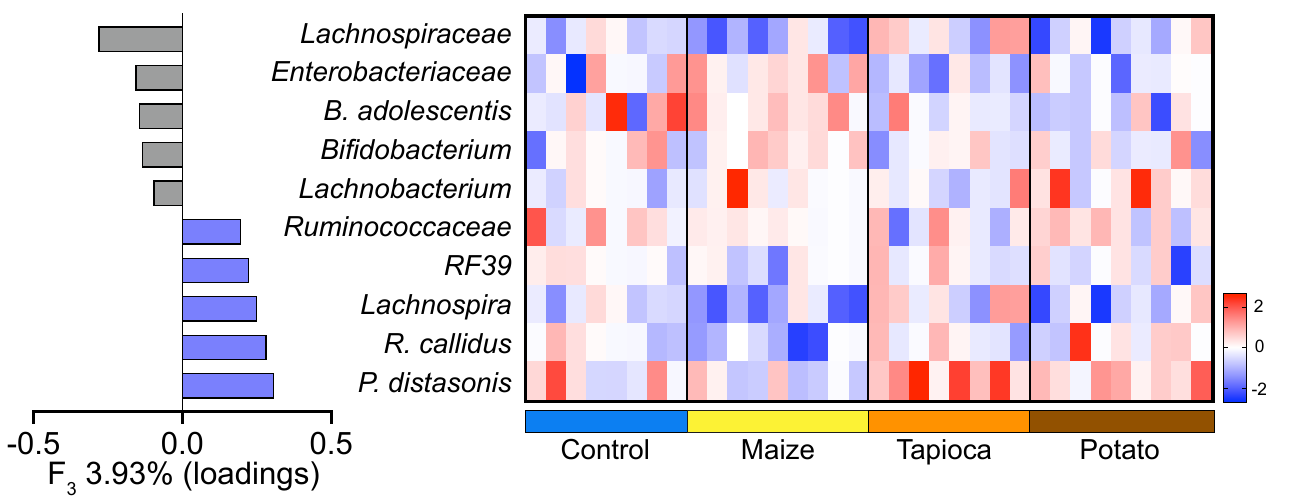}
    \end{subfigure}\hfill%
    \begin{subfigure}[t]{0.31\textwidth}
        \hfill
    \end{subfigure}
    \hspace*{-.25in}
    \begin{subfigure}[t]{.98\textwidth}
        \centering
        \caption{\label[cap]{fig:fibers.timeseries}}
        \includegraphics[scale=.95]{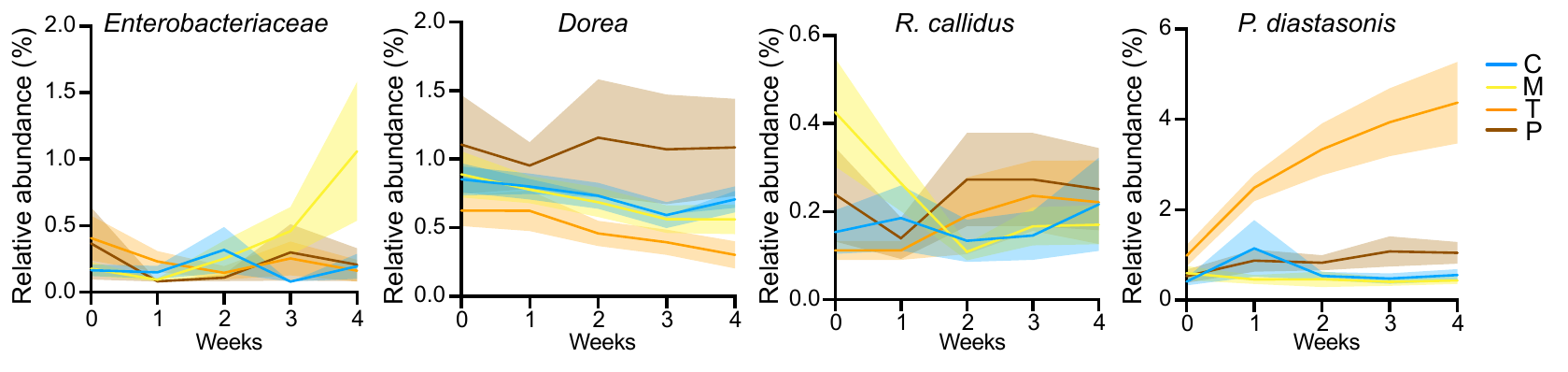}
    \end{subfigure}\hfill

    \caption{{\bf Comparison of \tcam with existing tensor based methods.}
    \Cref{fig:fibers.gemelli.factors} Ordination scatter plot of the data from~\cite{Deehan2020}, based on the Gemelli method~\cite{Martino2020} Inset: pairwise PERMANOVA.
    \Cref{fig:fibers.gemelli.funnel} Univariate statistical testing funnel comparison between Gemelli (left) and \tcam (right) pruning strategies.
    \Cref{fig:fibers.tcam} \tcam ordination graph for data of~\cite{Deehan2020}, ; Inset: Pairwise PERMANOVA
    \Cref{fig:fibers.heatmap} Barplot with top and bottom 2.5\% loadings for F3 (left). Heatmap representing per-subject AUC (log scale) for F3 top and bottom 2.5\% contributing bacteria (right).
    \Cref{fig:fibers.timeseries} Time series graphs describing all significant bacteria when applying \tcam pruning strategy. }
\end{figure}

%% file: child_docs/figure_panels/iris_figure.tex
\begin{figure}[H]
    \centering
    \begin{subfigure}[t]{0.24\textwidth}
        \centering
        \caption{\label[cap]{fig:iris.tcam}}
        \includegraphics[scale=.8]{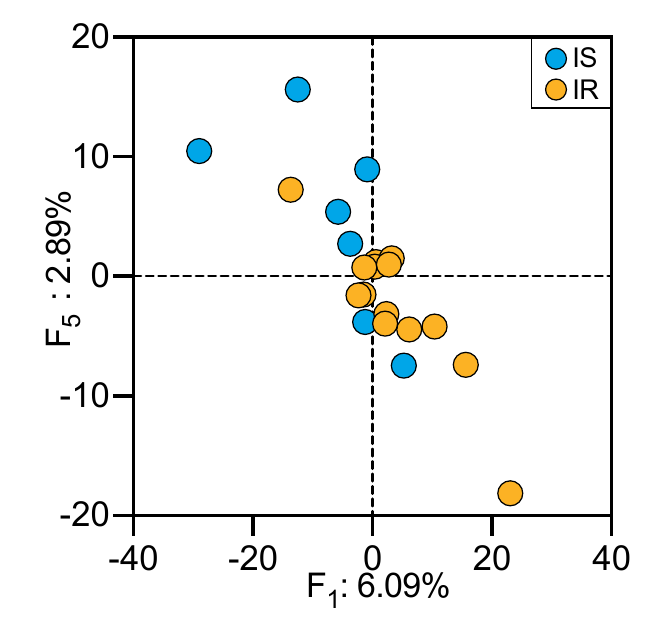}
    \end{subfigure}\hfill%
    \begin{subfigure}[t]{0.6\textwidth}
        \centering
        \caption{\label[cap]{fig:iris.heatmap}}
        \includegraphics[scale=.8]{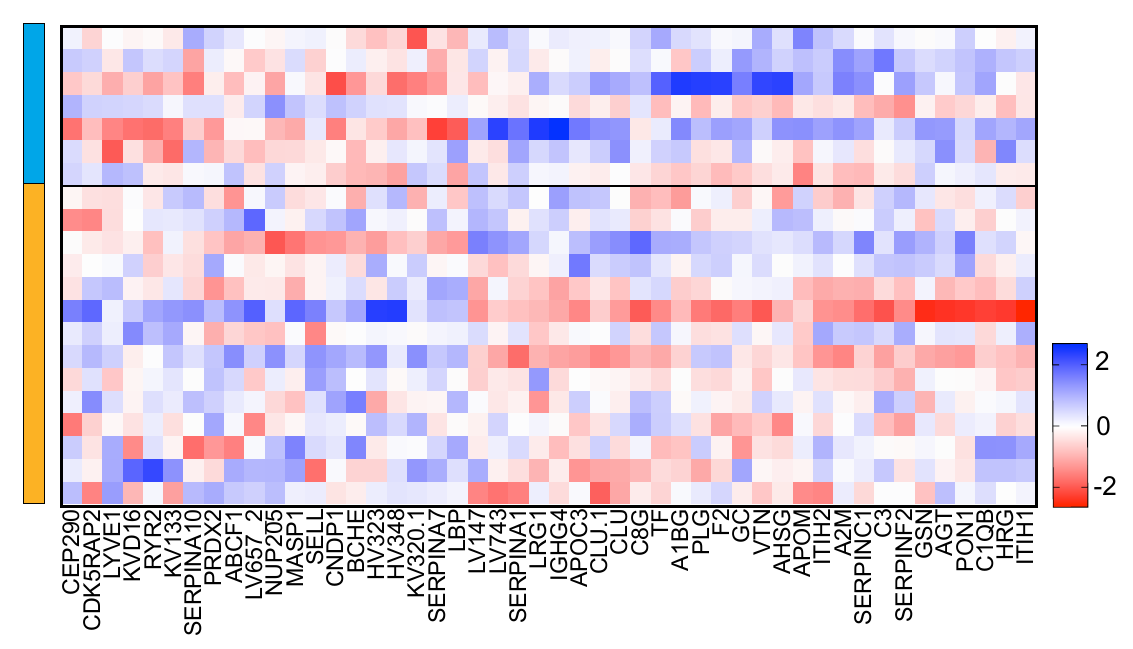}
    \end{subfigure}
    \caption{{\bf \tcam{}'s applicability to proteomics datasets.} 
    \Cref{fig:iris.tcam} Scatterplot of leading \tcam factors significantly correlated with SSPG. Points are colored according to insulin resistant (IR) and insulin sensitive (IS) information.
    \Cref{fig:iris.heatmap} Heatmap showing the sum of top and bottom 25 features contributing to the variation on \(F_1\) according to their loadings.}
\end{figure}

%% file: child_docs/figure_panels/ibd_figure.tex
\begin{figure}[H]
    \hspace*{-.25in}
    \centering
    \begin{subfigure}[t]{.27\textwidth}
        \centering
        \caption{\label[cap]{fig:ibd.tcam.roc}}
        \includegraphics[trim={0 0 4mm 5mm},clip,scale=.8]{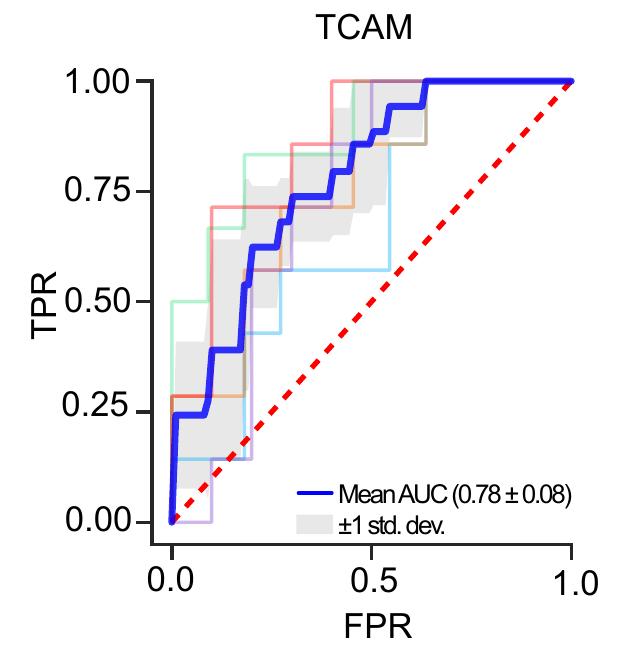}
    \end{subfigure}\hfill%
    \begin{subfigure}[t]{.35\textwidth}
        \centering
        \caption{\label[cap]{fig:ibd.tcam.important}}
        \includegraphics[scale=.81]{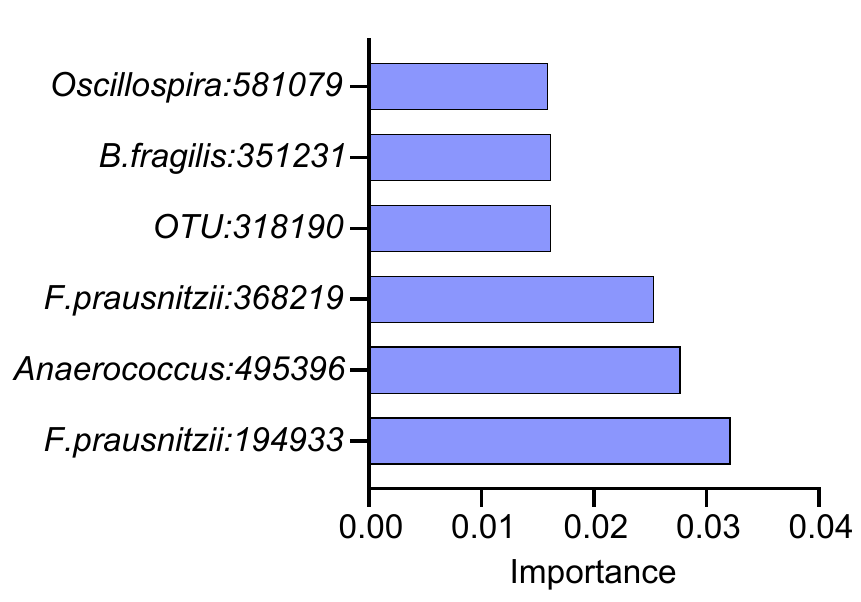}
    \end{subfigure}\hfill%
    \begin{subfigure}[t]{.27\textwidth}
        \vspace*{-2mm}
        \centering
        \caption{\label[cap]{fig:ibd.pruned.tcam}}
        \includegraphics[scale=.78]{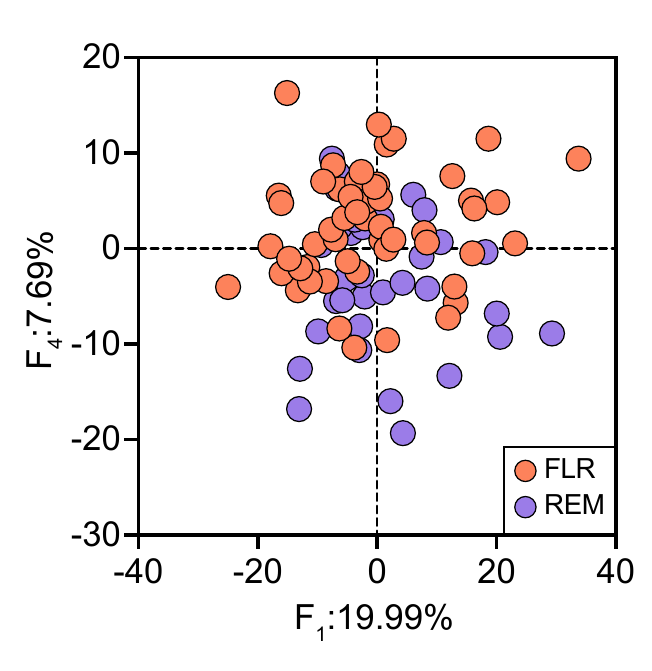}
    \end{subfigure}

    \caption{{\bf \tcam enables new discoveries and amenable for ML application.} 
    \Cref{fig:ibd.tcam.roc} ROC curve for MLP model trained to classify remission/flare based on \tcam{} transformed data of all timepoints.
    \Cref{fig:ibd.tcam.important} Bar plot showing importance scores of top 5\% ranked features.
    \Cref{fig:ibd.pruned.tcam} Scatterplot of \tcam{} scores computed on top 5\% most important features}
\end{figure}

%% file: child_docs/online_methods.tex
A real tensor of order-$N$, denoted by $\tA \in \RR^{d_1 \xx d_2 \xx \cdots \xx d_N}$, is a multi-dimensional array with real entries indexed by $N$-tuples. 
For example, the  $i_1 , i_2 , \dots , i_N$ entry of $\tA$ is denoted by $\tA_{i_1 , i_2 , \dots , i_N}$.
In this paper, we consider $3^{rd}$ order tensors $\tA \in \RR^{\mpn}$ holding data from $p$-dimensional samples, collected from $m$ subjects across $n$ time-points.
The size of $p$ is determined by the number of features measured in the 'omics method being used, it can be the number of observed bacterial species in metagenomics sequencing or the number of genes in transcriptomics etc.
We use Matlab notations for slicing and indexing of tensors, e.g., $\tA_{i,:,:} \in \RR^{1 \xx p \xx n}$ denotes the $i^{th}$ {\it horizontal slice} of $\tA$, which may be considered as a $p \xx n$ matrix. 

\vspace{0.2cm}

\phantomsection \label{def:datapreproc.mdf}
\noindent
The mean sample of a tensor $\tA$, is defined as $\bar{\tA} = 1/m \sum_{i=1}^{m} \tA_{i,:,:} \in \RR^{1 \xx p \xx n}$.
A tensor $\tA$ is in {\bf mean-deviation form} (MDF), if  $\FNormS{\bar{\tA}} = 0$ , where $\FNorm{\cdot}$ denotes the Frobenius norm: $\FNormS{\tA} = \sum_{k,j,i} \tA_{k,j,i}$. Any tensor $\tA$ can be centered to MDF by subtracting its mean sample from each horizontal slice. 

\noindent
Given a non-singular $n \xx n$ matrix $\matM$, the tubal singular value decomposition with respect to the the $\Mprod$-product (\tsvdm{}) of $\tA$ is written as $\tA = \tU \Mprod \tS \Mprod \tVt$ where $\Mprod$ denotes the tensor-tensor product (Refer to ~\cite{Kilmer} for original definitions, and~\cref{app:discussion} for details). Throughout this study, we considered $\matM$ defined by the discrete cosine transform (DCT). 

\noindent
Given a tensor $\tA$ in MDF, the \tcam{} of $\tA$ is defined by a {\it scores} matrix $\matZ \in \RR^{m \xx pn}$ whose $(\ell,h)$ entry is $ ( \widehat{\tA \Mprod \tV} )_{\ell, \matr_{h,1}, \matr_{h,2}}$, and a $np \xx p$ {\it loadings} matrix $\mat{V} $ with entries $\mat{V}_{h,j} = \thV_{\matr_{h,1},j,\matr_{h,2}}$.
The notation $\widehat{\tX}$ denotes the domain transform of a tensor $\tX$  and $\matr = \{ (\matr_{h,1}, \matr_{h,2}) \}_{h=1}^{pn}$ is an ordered collection of tuples such that $\thS_{\matr_{1,1} , \matr_{1,1}, \matr_{1,2}} \geq \thS_{\matr_{2,1} , \matr_{2,1}, \matr_{2,2}} \geq \cdots \geq \thS_{\matr_{np,1} , \matr_{np,1}, \matr_{np,2}}$ .
Each row of the factors matrix represents the $p$-dimensional time-series (trajectory) of each subject, while the loadings matrix measures the contribution - magnitude and direction - of each of the $p$ 'omics features to each of the \tcam{} factors across samples. 
Refer to ~\cref{app:discussion} for formal definitions, construction and mathematical optimality guarantees.

\noindent
Excluding MDF, the \tcam{} makes no assumptions on the data, making it suitable for any choice of normalization method. Unless stated otherwise, data were normalized to form log2 folds from baseline (\LFB);  for experiment with timepoints $t_1 , \dots , t_n$, where $t_1 ... t_j$ are  considered the baseline samples. Let $s_{\ell 1},...,s_{\ell n}$  denote the samples collected from subject $\ell$. The \LFB transformed data $\hat{s}_{\ell 1},...,\hat{s}_{\ell n}$  is defined by $ \hat{s}_{\ell k} \coloneqq \log_2 (s_{\ell k} / \bar{s_\ell})$  where $\bar{s_\ell}$  is the mean of subject $\ell$'s baseline samples: $s_{\ell 1},...,s_{\ell j}$.

\subsection*{Data processing}

All \tcam were computed on MDF of the data.  
The pre-processing and analysis steps taken vary between datasets and are listed below. 
Fine details are described in the provided code~\footnote{\href{https://github.com/UriaMorP/tcam_analysis_notebooks}{https://github.com/UriaMorP/tcam\_analysis\_notebooks}}. 

\paragraph*{Post antibiotics reconstitution}
Shotgun metagenomics sequencing data of stool samples was downloaded from ENA (project accession number: PRJEB28097). 
QC filtration and read trimming was done using fastp, followed by removal of reads mapped to human genome by bowtie2 mapper. MetaPhlan3 was used for taxonomic assignment of the reads. 
For the analysis, we included the timepoints with missing samples of no more than 2 subjects, allowing for single day deviation in any direction. 
Subjects without at least one sample in each phase of the experiment (baseline, antibiotics, intervention) were excluded. 
Relative abundance values were capped at \(10^{-6}\) and features with maximal values less then \(10^{-6}\) were omitted.  
For the analysis using \tcam{}, data from each participant were \LFB transformed. 
PERMANOVA was computed using truncated distance matrices reconstructed using the minimal number of components (either \tcam or PCA) such that the truncation accounts for at least 20\% of the total variation in the data. 
Feature selection for univariate time-series analysis was done by taking the 0.75 quantile of loadings norm computed for the factors demonstrating significant univariate difference (ANOVA). 
Univariate time-series analysis was performed using lmer. 

\paragraph*{Dietary fiber intervention}
16S rDNA sequencing data and metadata from this study~\cite{Deehan2020} were downloaded from ENA (project accession number: PRJNA560950). 
Overlapping paired-end FASTQ files of 16S amplicon sequencing data were matched and analyzed using the Qiime2 pipeline (q2cli version 2021.4.0)~\cite{Bolyen2019}. 
Poor quality bases were trimmed, sequences were denoised and binned to amplicon sequence variants (ASVs) using the dada2 plugin for Qiime2~\cite{Callahan2016}.
Taxonomic assignment was performed using naive Bayes feature classifier and Greengenes 13\_8 database. 
Gemelli method~\cite{Martino2020} was run using raw counts and default parameters. 
For the analysis using \tcam method, relative abundance values were capped at \(10^{-3}\), features with maximal values less then \(10^{-3}\) were omitted and data from each participant were \LFB transformed. 
PERMANOVA was computed using distance matrices reconstructed with all components (either \tcam or PCA). 
Feature selection for univariate time-series analysis was done by taking the 0.75 quantile of loadings norm computed for the factors demonstrating significant univariate difference (ANOVA). 
Univariate time-series analysis was performed using lmer. 

\paragraph*{Pediatric ulcerative colitis}
Sample-specific metadata and final microbial OTU relative abundances were acquired from~\cite{Schirmer2018}. Subjects without at least one sample in each timepoint of the experiment (0, 4, 12, 52) were excluded. 
Relative abundance values were capped at \(10^{-4}\), features with maximal values less then \(10^{-4}\) were omitted and data from each participant were \LFB transformed.
To study differences in temporal microbiome composition between the three treatment groups, PERMANOVA was computed using distance matrices reconstructed with all components. 
A multilayer perceptron (MLP) with single 1000 neuron wide hidden layer was trained to predict treatment groups (5ASA, CS-Oral and CS-IV) using the minimal number of \tcam factors such that at least 90\% of the variation in the data is explained by the factors. 
On the other hand,a similar MLP architecture was trained to predict the combined treatment groups (5ASA, CS), using the minimal number of \tcam factors such that at least 80\% of the variation in the data is explained by the factors. 
Similarly, an MLP with single 1000 neuron wide hidden layer was trained to predict remission state using the minimal number of \tcam factors such that at least 80\% of the variation in the data is explained by the factors. 
The same MLP architecture was trained using log2 fold changes of weeks 12 and 52 to the baseline.

\paragraph*{Identification of insulin resistance using longitudinal proteomics}
Proteomics data and metadata were downloaded from ~\href{https://figshare.com/articles/dataset/Multi\_Omics\_Seasonal\_RData/12376508}{https://figshare.com/articles/dataset/Multi\_Omics\_Seasonal\_RData/12376508}. 
Three years duration was stratified to trimesters, and repeated samples of the same participant within a trimester were median aggregated. 
Subjects lacking measurements in more than a single trimester were omitted from the analysis. 
Missing timepoints filled via linear interpolation or forward/backward filled in case of missing last/first trimester. 
The first year was considered as ``baseline". 
Since proteomics data contain negative values, the data for each subject was shifted by the median baseline measurements for that subject.


%% file: child_docs/figure_panels/S1.tex
\begin{supfigure}[ph!]
    \hspace*{-.25in}
    \centering
    \begin{subfigure}[b]{0.31\textwidth}
        \centering
        \caption{\label[sfig]{sfig:postabx.pca.baseline}}
        \includegraphics[scale=.8]{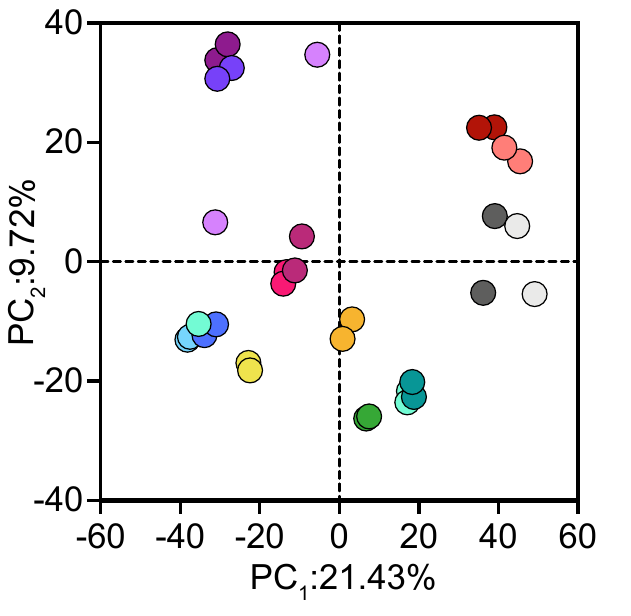}
    \end{subfigure}%
    \begin{subfigure}[b]{0.31\textwidth}
        \centering
        \caption{\label[sfig]{sfig:postabx.pca.group}}
        \includegraphics[scale=.8]{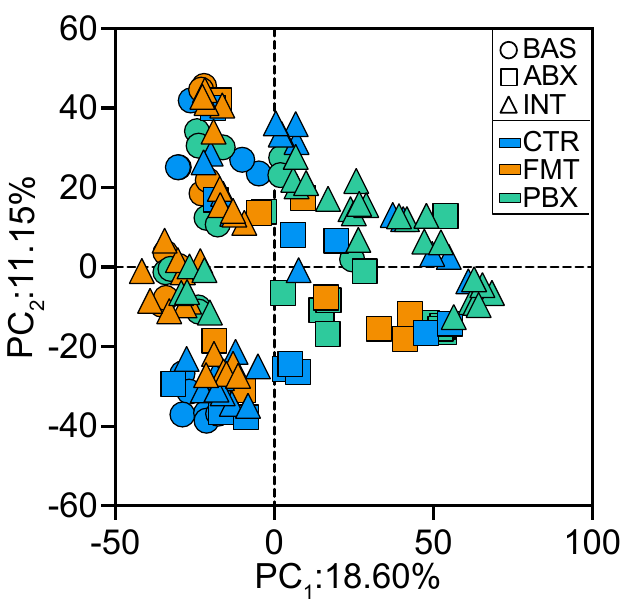}
    \end{subfigure}%
    \begin{subfigure}[b]{0.31\textwidth}
        \centering
        \caption{\label[sfig]{sfig:postabx.pca.group.BAS}}
        \includegraphics[scale=.8]{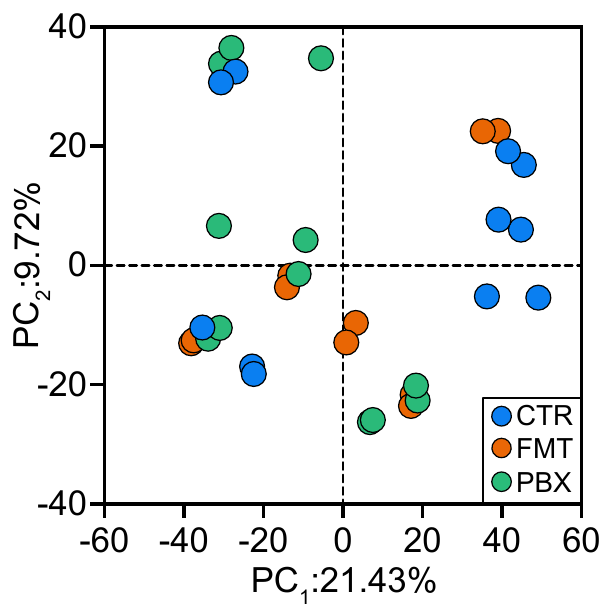}
    \end{subfigure}
    \begin{subfigure}[b]{0.31\textwidth}
        \centering
        \caption{\label[sfig]{sfig:postabx.pca.group.ABX}}
        \includegraphics[scale=.8]{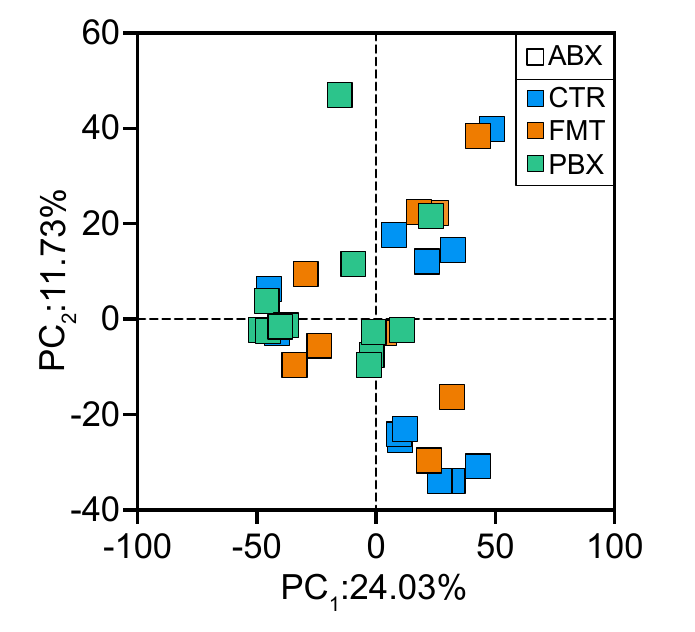}
    \end{subfigure}%
    \begin{subfigure}[b]{0.31\textwidth}
        \centering
        \caption{\label[sfig]{sfig:postabx.pca.group.INT}}
        \includegraphics[scale=.8]{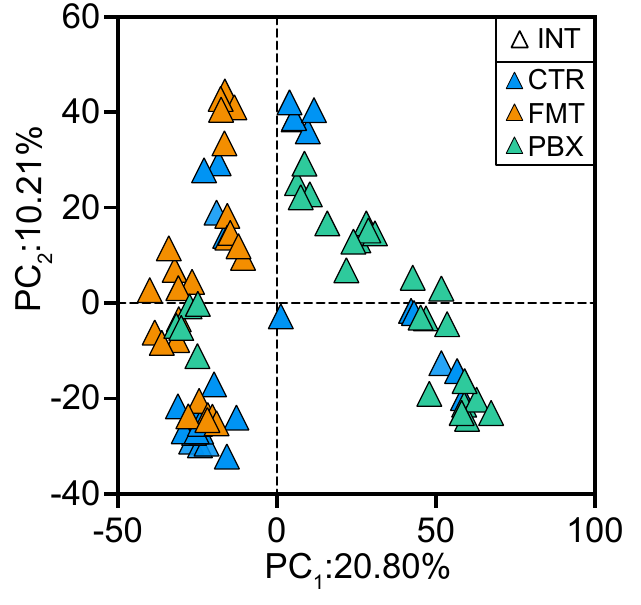}
    \end{subfigure}\hfill%
    \begin{subfigure}[b]{0.31\textwidth}
    \hfill
    \end{subfigure}

  \caption{{\bf Comparison of \tcam with existing matrix based methods for exploratory analysis }
  \Cref{sfig:postabx.pca.baseline} PCA plot of baseline timepoints, 1-2 samples per each subjects. Points are colored according to participant.
  \Cref{sfig:postabx.pca.group} PCA plot of all timepoints. Points are colored according to group.
  \Cref{sfig:postabx.pca.group.BAS,sfig:postabx.pca.group.ABX,sfig:postabx.pca.group.INT} PCA plot of baseline, antibiotics, and intervention phases respectively. Points are colored according to group}

\end{supfigure}

\begin{supfigure}
    \begin{subfigure}[b]{1\textwidth}
        \centering
        \caption{\label[sfig]{fig:postabx.timeser.tcam}}
        \includegraphics[scale=.65]{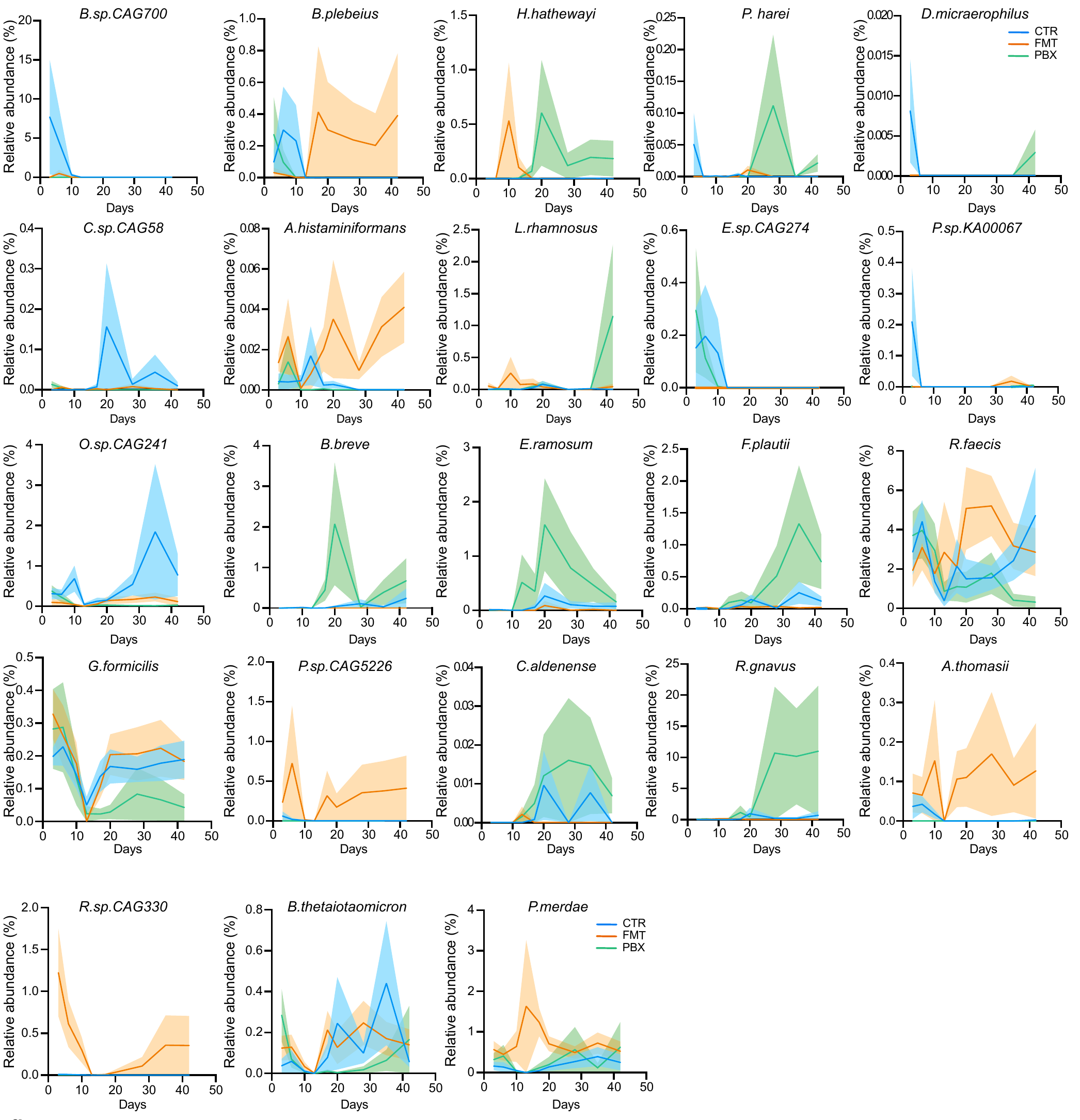}
    \end{subfigure}
    \begin{subfigure}[b]{1\textwidth}
      \centering
      \caption{\label[sfig]{fig:postabx.timeser.mutual}}
      \includegraphics[scale=.65]{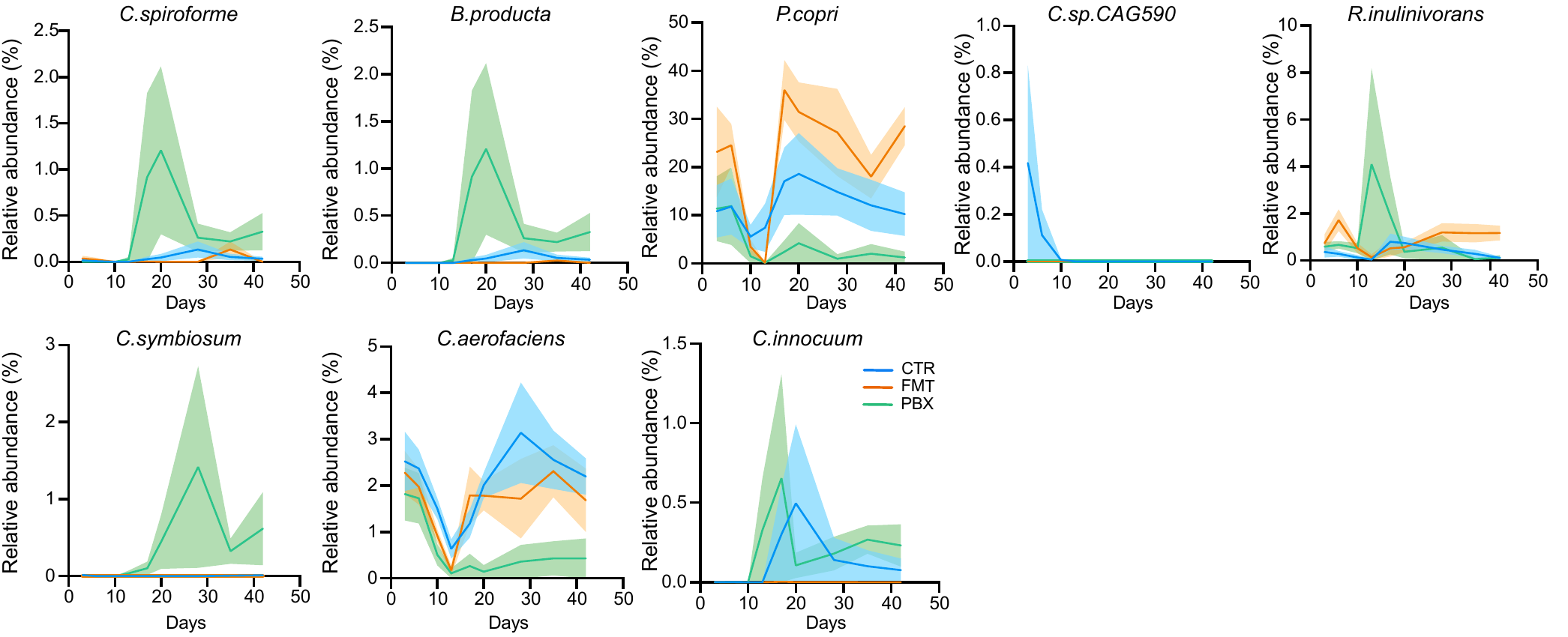}
    \end{subfigure}
    \begin{subfigure}[b]{1\textwidth}
        \centering
        \caption{\label[sfig]{fig:postabx.timeser.all}}
        \includegraphics[scale=.65]{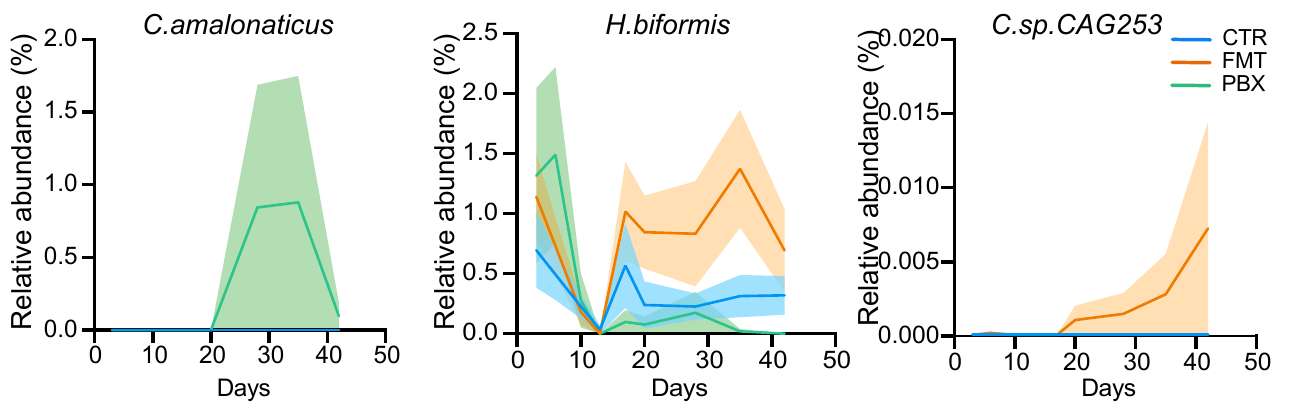}
    \end{subfigure}
    
    \caption{{\bf Comparison of discovery rates between naive time-series analysis and \tcam based pruning}
    \Cref{fig:postabx.timeser.tcam} Time series of relative abundance levels for features discovered only when pruning the features.
    \Cref{fig:postabx.timeser.all,fig:postabx.timeser.mutual} Time series of relative abundance levels for features discovered when no pruning scheme is used (top) and by both methods (bottom).}    
\end{supfigure}
\pagebreak

%% file: child_docs/figure_panels/S2_IBD.tex
\begin{supfigure}[ph!]
    \hspace*{-.25in}
    \centering
    \begin{subfigure}[t]{.31\textwidth}
        \centering
        \caption{\label[sfig]{fig:ibd.remission.projections.tcam}}
        \includegraphics[scale=.8]{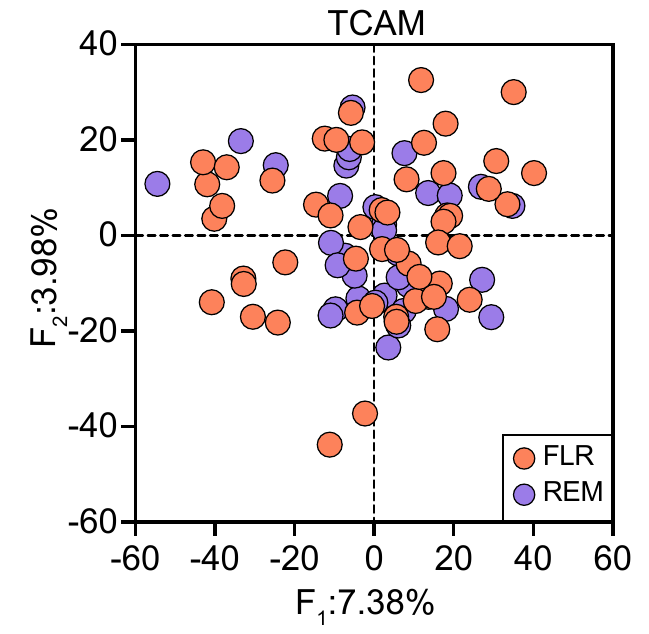}
    \end{subfigure}\hfill%
    \begin{subfigure}[t]{.31\textwidth}
        \centering
        \caption{\label[sfig]{fig:ibd.remission.projections.w12}}
        \includegraphics[scale=.8]{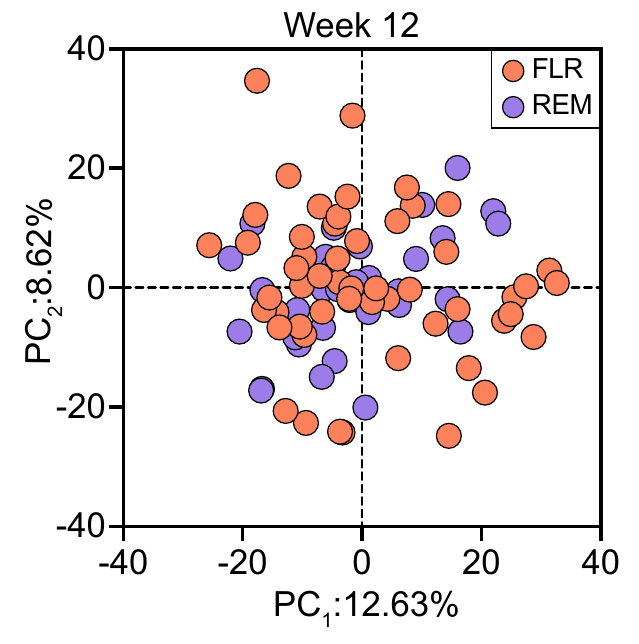}
    \end{subfigure}\hfill%
    \begin{subfigure}[t]{.31\textwidth}
        \centering
        \caption{\label[sfig]{fig:ibd.remission.projections.w52}}
        \includegraphics[scale=.8]{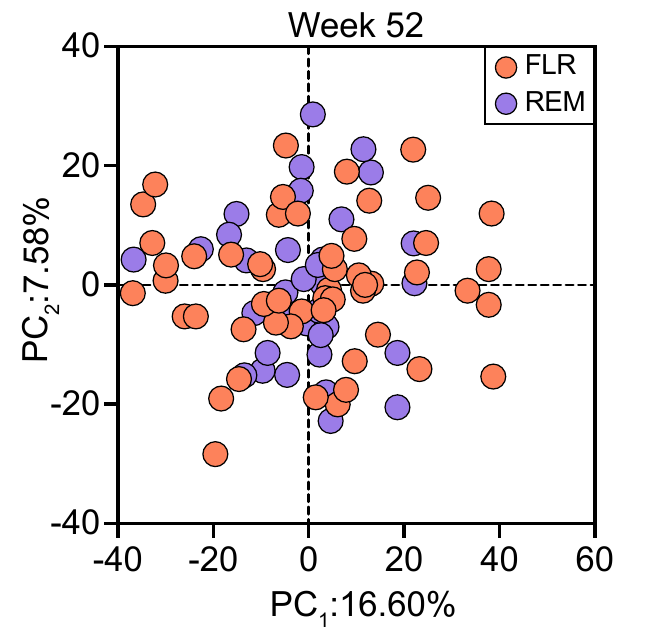}
    \end{subfigure}
    
    \hspace*{-.25in}
    \centering
    \begin{subfigure}[t]{1\textwidth}
        \centering
        \caption{\label[sfig]{fig:abundances.important.features}}
        \includegraphics[scale=.9]{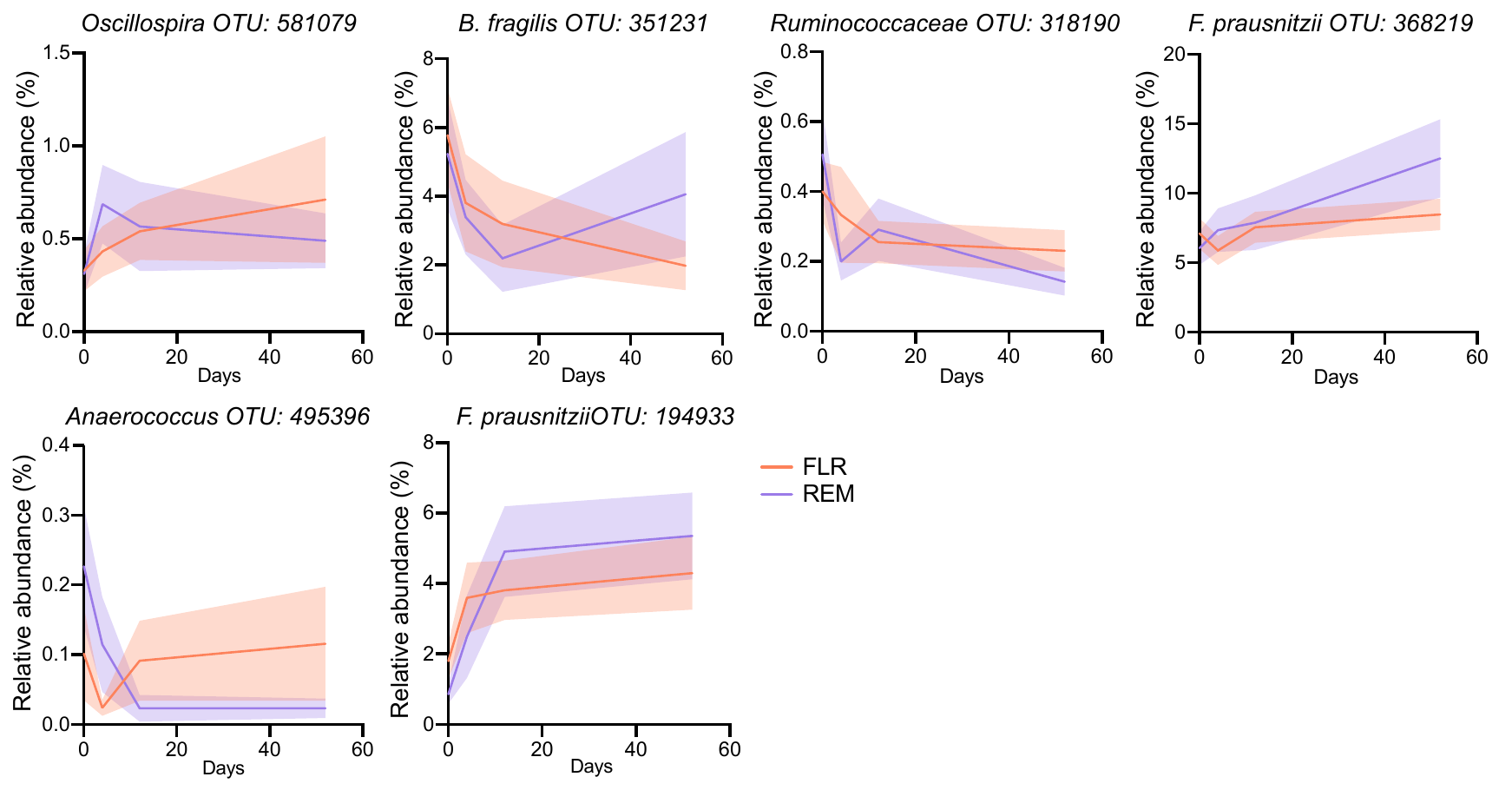}
    \end{subfigure}
    
    \hspace*{-.25in}
    \centering
    \begin{subfigure}[t]{.31\textwidth}
        \centering
        \caption{\label[sfig]{fig:ibd.roc.w12}}
        \includegraphics[scale=.8]{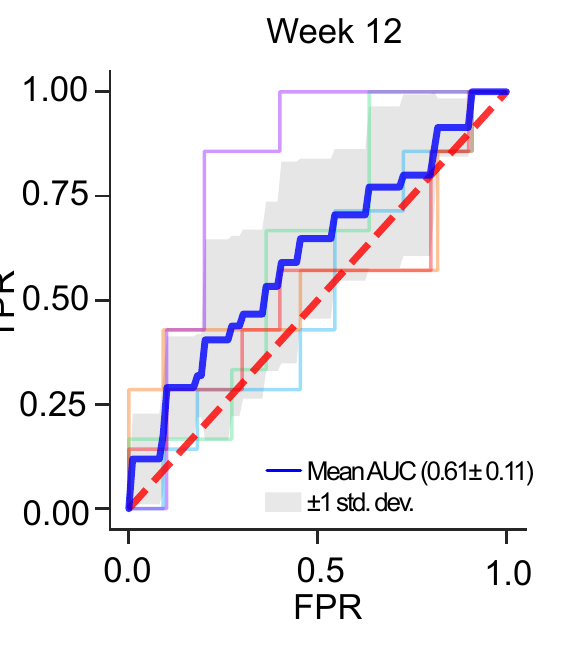}
    \end{subfigure}\hfill%
    \begin{subfigure}[t]{.31\textwidth}
        \centering
        \caption{\label[sfig]{fig:ibd.roc.w52}}
        \includegraphics[scale=.8]{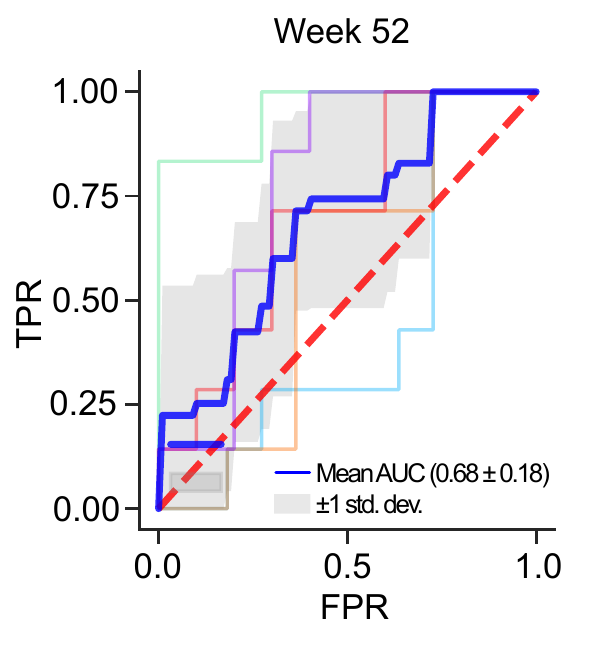}
    \end{subfigure}\hfill%
    \begin{subfigure}[t]{.331\textwidth}
        \hfill
    \end{subfigure}
    
    \caption{{\bf \tcam analysis on pediatric UC patients.}
    \Cref{fig:ibd.remission.projections.tcam,fig:ibd.remission.projections.w12,fig:ibd.remission.projections.w52} Scatter plots for 2 leading factors of  \tcam{} for the whole dataset (\Cref{fig:ibd.remission.projections.tcam}); PCA computed for log2 ratio of week 12 and baseline (\Cref{fig:ibd.remission.projections.w12}); PCA computed for log2 ratio of week 52 and baseline (\Cref{fig:ibd.remission.projections.w52}). Points are colored according to remission (REM) and flare (FLR) status.
    \Cref{fig:abundances.important.features} Time series of relative abundance levels, highlighting the differences in trajectories of the features contributing to the remission status classification model.
    \Cref{fig:ibd.roc.w12,fig:ibd.roc.w52} ROC curve for MLP model trained to classify remission/flare based on PCA transformed log fold change between week 12 and baseline (~\cref{fig:ibd.roc.w12}) and log fold change between week 52 and baseline (~\cref{fig:ibd.roc.w52}).
    (~\hyperref[par:ibd]{back to text})}    
\end{supfigure}
  
  

%% file: child_docs/supp_discussion.tex
A real order-\(N\) tensor \(\tA \in \RR^{d_1 \xx d_2 \xx ... \xx d_N}\)  is an multi-dimensional array of entries\footnote{In pure mathematics this is actually the definition of a {\em hypermatrix}, while a {\em tensor} is an algebraic object that describes a multilinear relationship. However, in data science often the term `tensor' is abused to mean a hypermatrix, and we adopt this terminology here as well.}. Each of the entries of \(\tA\) can be referred to by specifying an \(N\)-tuple of numbers $(i_1 ,..., i_N)$ where \(i_k \in \{ 1, 2,\dots,d_k \} \). 

A third-order tensor  \(\tA \in \RR^{\mpn}\) can also be viewed as an \(m\) elements long list of \(p \xx n\) matrices, each an horizontal slice of the tensor (\Cref{sfig:cartoon.horizontal.slices.full,sfig:cartoon.horizontal.slices}). This mathematical construct is appealing in the context of  longitudinal studies as it enables storing the data in a way that is consistent with the data collection. One might think of \(\tA\) as a data-structure for holding the results of an experiment during which \(n\) samples were collected from \(m\) participants, and each sample is characterized by \(p\) features. These features may be genes in the case of RNA-seq samples, taxonomic composition of shotgun metagenomics sequencing, etc.
The tensor data structure reflects not only the data points but also key relationships between them. 
For example, let \(1 \leq i \leq p \) an integer denoting the index of a certain feature, then variations of this feature across the whole cohort are obtained by fixing the second coordinate of the tensor to \(i\): \(\tA_{:,i,:}\)\footnote{Here, we are using {\sc Matlab} notation, in which `:' denotes the entire range of a mode, and `\(l:k\)' denotes \(\{l, l+1, \dots, k\}\).}. Similarly, tracking this feature in a single timepoint \(t\) is done by restriction of the last two coordinates \(\tA_{:,i,t}\).
\input{child_docs/figure_panels/sup_cartoons}
The arrangement of the same data in the form of a matrix, which has only two dimensions, would require us to make a somewhat arbitrary choice about which of the two dimensions are to be coalesced into a single dimension. 
For example, one might consider each individual subject as a single sample, and as such each individual will have a designated row in the data matrix, while  repeated measurements of the same feature at several timepoints are treated as entirely different features, resulting in an \(m \xx np\) matrix. Note that by concatenating the repeated samples of each individual, we form a new feature space in which the temporal context is lost. 
Another option is to define samples as \(p\) dimensional entities measured for each subject at all timepoints, resulting in an \(nm \xx p\) matrix. This formulation breaks the correspondence between data of the same individual across timepoints, as well as the data of all individuals at a single timepoint. 

Continuing with the above example, if in addition we were to sample each subject at \(k\) body sites (instead of one), then the number of possible ways of rearranging data in the form of a matrix would have increased to four different choices, each of which captures a different aspect of the experimental design. On the other hand, having the data in a format of tensor would only require an addition of a single mode for describing the \(k\) different body sites, resulting in a fourth-order tensor in \(\RR^{\mpn \xx k}\). 
The advantage of representing data in native tensor format as opposed to {\em matricizing} it is not unique to longitudinal studies and appears in many domains. See~\cite{Kilmer} for discussion.

In this work, we are concerned with tensor-based dimensionality reduction methods. It was previously shown, e.g.~\cite{Williams2018} and more recently~\cite{Martino2020,Delannoy-Bruno2021}, that tensor-based dimensionality reduction methods have a potential to provide more meaningful output compared to analogous matrix-based methods. Most of these works consider  the low-rank approximation in the form of a CP factorization~\cite{Hitchcock1927}
\begin{equation}\label{app:supdisc.cp.def}
	\tA \approx \tA_r = \sum_{i=1}^{r} \sigma_i \u_{i}^{(1)} \circ \u_{i}^{(2)} \circ ... \circ  \u_{i}^{(N)} \in \RR^{d_1 \xx ... \xx d_N} ,
\end{equation}
where \(\sigma_i \) are positive scalers, the {\it components} \(\u_{i}^{(j)} \) are \(d_j\) dimensional unit vectors and \(\circ\) denotes the outer-product operation (\(\u_{i}^{(j)} \circ \u_{i}^{(k)}  \) is a \(d_j \xx d_k\) matrix, while \(\u_{i}^{(j)} \circ \u_{i}^{(k)} \circ \u_{i}^{(l)}  \) is a tensor in \(\RR^{d_j \xx d_k \xx d_l }\)). 
The number of terms \(r\) in the above factorization denotes the maximal rank of the sought approximate \(\tA_r\) of the original tensor \(\tA\). For brevity, we denote a CP factorization by \(\tA_r = [\bf{\Sigma}; \matU_1 ,..., \matU_N]\) where \(\bf{\Sigma} = \text{diag} (\sigma_1,...,\sigma_r) \), and \(\matU_j\) are \(d_j \xx r\) matrices of columns that are the unit vectors \(\u_{1}^{(j)},...\u_{r}^{(j)}\). As a convention, each of the \(r\) summands 
($\sigma_i \u_{i}^{(1)} \circ ... \circ  \u_{i}^{(N)}$)
 is a rank-1 tensor (for alternative definitions of tensor rank see~\cite{Kolda2009}). 
This factorization provides an intuitive breakdown of the data which is somewhat analogous to that of a PCA for matrix data: the overall contribution of each component \(\u_{i}^{(j)}\) to the approximation is determined by the magnitude of the corresponding scaler \(\sigma_i\)  (larger scaler implies greater contribution), and the components themselves may reflect the different modalities of the data (in the above example, \(\u_{i}^{(1)}\) are associated with the different subjects, while \(\u_{i}^{(2)},\u_{i}^{(3)}\) components are associated with features and timepoints respectively). 

Computing an approximationin n the form of a CP decomposition of a given tensor $\tA$ is usually accomplished by solving the following optimization problem:
\begin{eqnarray}
\label{eq:opt-cp}
	[\bf{\Sigma}; \matU_1 ,..., \matU_N] &= \argmin_{\tilde{\bf{\Sigma}}, \tilde{\matU}_1 ,..., \tilde{\matU}_N} \FNormS{\tA - \sum_{i=1}^{r} \tilde{\sigma}_i \tilde{\u}_{i}^{(1)} \circ  ... \circ  \tilde{\u}_{i}^{(N)} }
\end{eqnarray} 
where $\FNorm{\tX}$ denotes the Frobenius norm of the tensor $\tX$: 
$$
\FNormS{\tX} = \sum_{i_1 = 1}^{d_1} \cdots \sum_{i_N = 1}^{d_N} \tX_{i_1 , i_2 , \dots i_N} ^2
$$
The current gold-standard method for solving this problem is Alternating Least Squares (ALS)~\cite{Kolda2009}.

As discussed in the main text, it is generally hard to solve Problem~\eqref{eq:opt-cp} as it is non-convex, thus potentially having many local minimizers which are not global minima. Indeed, even the gold-standard method (ALS) is not guaranteed to find a global minimum.  Moreover, given a data tensor \( \tA \) (subjects, features and time),  and its CP form approximate \(\tA_r = [\bf{\Sigma}; \matU_1 ,\matU_2, \matU_3] \),  suppose that we are provided with data for a new participant \(\tX \in  \RR^{1 \xx p \xx n}\), then the task of extending the current factorization to the new data tensor, which is \(\tA\) augmented with the new sample \(\tX\), is far from a trivial one. 
Such out-of-sample extension capability is a fundamental requirement from any embedding algorithm  that we wish to use as a step in a ML pipeline. 

The aim of \tcam is to provide a tensor-based PCA-like tool that is better  than CP, in terms of ease of interpretation of the model's outcomes and mathematical properties that ensure safe application to downstream statistical analysis and ML workflows.  


%% file: child_docs/figure_panels/sup_cartoons.tex
\begin{supfigure}[ph!]
    \hspace*{-.25in}
    \centering
    \begin{subfigure}[b]{0.31\textwidth}
        \centering
        \caption{\label[sfig]{sfig:cartoon.horizontal.slices.full}}
        \includegraphics{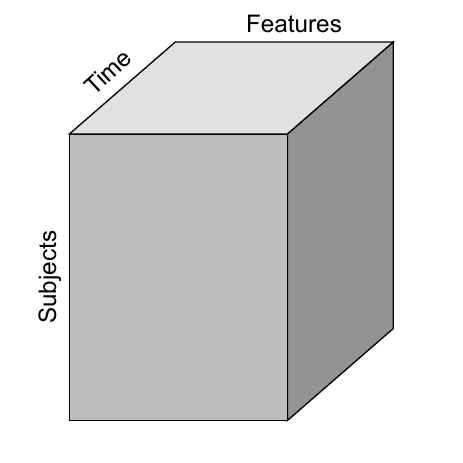}
    \end{subfigure}
    \begin{subfigure}[b]{0.31\textwidth}
        \centering
        \caption{\label[sfig]{sfig:cartoon.horizontal.slices}}
        \includegraphics{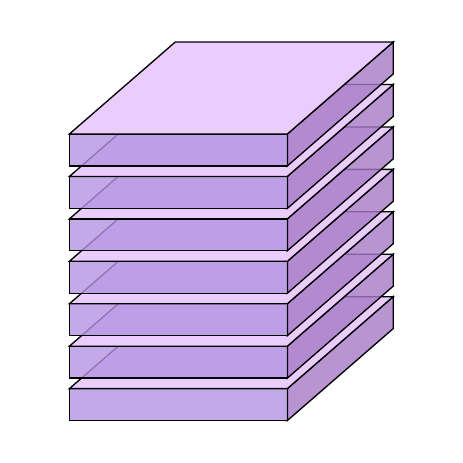}
    \end{subfigure}
    
    \caption{{\bf Subject centered view of $3^{rd}$ order tensor }. 
    \Cref{sfig:cartoon.horizontal.slices.full} An illustration of the data structure. 
    \Cref{sfig:cartoon.horizontal.slices} The right panel presents a breakdown of the left tensor into $m$ horizontal slices that are $p \xx n $ matrices.}
\end{supfigure}

%% file: child_docs/supp_pca.tex
Since we strive to have to claim that \tcam is `PCA-like', it is useful to first present a brief definition of PCA, and discuss it characteristics. All constructions, definitions and properties presented in this section are taken from~\cite{jolliffe1986principal}.

Let \(\matA \in \RR^{m \xx p}\) be a data matrix with rows \(\mat{a}_1, \dots, \mat{a}_m\) corresponding to \(m\) samples, and assume that \(\matA\) has been centered (so the column means are zero). PCA is defined by an orthogonal linear transformation \(\matW \) transforming the rows of \(\matA\) to a new coordinates system, in which, the largest portion of variation in the data lies on the first axis (called the first principal component of \(\matA\)), i.e. 
$
    \z_1 = \matA \w_1
$
where $\w_1 \in \RR^{p}$ is the maximizer of $\wt \matAt \matA \w$ subjected to $\TNormS{\w} = 1$.
The \(k^{th}\) largest portion of variation lies on the \(k^{th}\) principal component of \(\matA\), i.e. 
\begin{align*}
    &\z_k = \matA \w_k ,&
    &\w_k = \text{argmax}_{\TNorm{\w} = 1} \w^{\T} {\matA^{(k)}}^{\T} \matA^{(k)} \w 
\end{align*}
where given the first \(k-1\) principal components of \(\matA\) we define \(\matA^{(k)} = \matA (\mat{I} - \sum_{i=1}^{k-1} \w_i \w_i^{\T}) \). 
The complete factorization yields the expression $\matZ = \matA \matW$, 
where rows of the matrix \(\matZ \in \RR^{m \xx p}\)  are called the {\it sample PC scores} (the \(i^{th}\) row of \(\matZ\) contains the sample PC scores of the \(i^{th}\) sample \(\mat{a}_i\)). The orthogonal \(p \xx p\) matrix \(\matW\) is called the weight or coefficients matrix. 
The rank \(k\)-truncated PCA of a $m \xx p $ matrix $\matA$ is defined by
\begin{equation}\label{app:pca.trunc.def}
    \matZ_k = \matA \matW_k \in \RR^{m \xx k}
\end{equation}
where \(\matW_k \coloneqq [\w_1 , ... , \w_k]\) is an \(p \xx k\) matrix with orthonormal rows; \(\matW_{k}^{\T} \matW_{k} = \mat{I}_k\). The \(k\)-truncated PCA \(\matZ_k\) consists of only the first \(k\) principal components.

The {\bf sample variance-covariance matrix} of the collection \(\{\mat{a}_i\}_{i=1}^{m}\) of \(m\), (centered) \(p\)-dimensional samples \(\matS\) is given by
$
    \matS = (m-1)^{-1} \matA^{\T} \matA ~.~
$


\begin{property}[{\cite[Chapter~2,~Properties~A1~and~A2]{jolliffe1986principal}}]\label{app:pca.prop.max.var}
    For any integer \(q = 1, \dots, p \), consider  the transformation $\y_i = \matBt \mat{a}_i$, where \(\y \in \RR^{q}\) and \( \matB \) is a ${p \xx q}$ matrix with $q$ orthonormal columns.  
    Define \(\matS_{\y} = \matBt \matS \matB \in \RR^{q \xx q}\) the variance-covariance matrix for \(\y\).
    The variance component of $\matS_{\y}$, quantified by $\trace{\matS_{\y}}$, is maximized when taking $\matB =\matW_q$ and minimized when $\matB = [\w_{p-q} ,\w_{p-q+1}, ... \w_p]$
\end{property}


\begin{property}[{\cite[Chapter 3, Eq. 3.1.4, Property A3]{jolliffe1986principal}}]\label{app:pca.prop.spectral}
    $\matS = \sum_{i=1}^{p} \lambda_i \w_i \wt_i$
\end{property}

\begin{property}[{\cite[Chapter 3, Property G4]{jolliffe1986principal}}]\label{app:pca.prop.distortion}
Let \(\tilde{\matA} = [\mata_1 ;...;\mata_m] \) be a  \(m \xx p\) matrix  of \(m\), \(p\)-dimensional observations. Define \(\matA\) to be the \(m \xx p\) matrix whose \(i^{th}\) row is \(\mata_i - \bar{\mata} \) where \(\bar{\mata} = \frac{1}{m} \sum_{j=1}^{m} \mata_j\) and consider the matrix \(\matA \matAt \in \RR^{m \xx m} \). 
The \(i^{th}\) diagonal element of \(\matA \matAt\) is the squared Euclidean distance of the sample \(\mata_i\) from the point \(\bar{\mata}\) that is the center of gravity of the points \(\mata_1 ... \mata_m\). 
Also, the \(i,j\) entry of \(\matA \matAt\), given by $\langle \mata_i - \bar{\mata}, \mata_j - \bar{\mata} \rangle$, 
is the cosine of the angle between the lines joining the points \(\mata_j\) and \(\mata_i\) to \(\bar{\mata}\), scaled by the distances of \(\mata_j\) and \(\mata_i\) from \(\bar{\mata}\).  Suppose that \(\mata_1 ... \mata_m\) are projected to a \(q\)-dimensional subspace using a linear orthogonal transformation \(\y_i = \matBt \mata_i\), and let \(\matY \coloneqq [\y_1 - \bar{\y} ; ... ; \y_m - \bar{\y}] \in \RR^{m \xx q}\) for \(\bar{\y} = m^{-1} \sum_{j=1}^{m} \y_j \in \RR^{q}\).
Then the choice \(\matB = \matW_q\) minimizes the distortion \(\FNormS{\matY \matYt - \matA \matAt }\)
\end{property}

%% file: child_docs/supp_pca_svd.tex
One practical way of describing and computing PCA for a matrix \(\tilde{\matA} \in \RR^{m \xx p}\), is using the singular value decomposition of its column centered form. Let again \(\matA\) denote the column centered form of \(\tilde{\matA}\),. and let
$
    \matA = \matU \Sigma \matVt
$
be its Singular Value Decomposition (SVD). That is, $\matU\in\RR^{m \xx m}$ and  $\matV\in\RR^{p \xx p}$ are  orthonormal matrices, and \(\Sigma \in \RR^{m \xx p}\) is a matrix with non-negative elements \(\sigma_1 \geq \sigma_2 \geq ...\) on its main diagonal and zeros elsewhere. 
The sample variance-covariance matrix $\matS$ can rewritten as \(\matS = (m-1)^{-1}  \matV \Sigma^{\T}\Sigma \matVt\), thus the columns of \(\matV\) are the eigenvectors of \(\matS\) and, according to~\Cref{app:pca.prop.spectral}, the coefficients in the PCA decomposition of \(\tilde{\matA}\).

Let us denote the rank-\(r\) truncated SVD of \(\matA\) by 
$
    \matA_r = \matU_r \Sigma_r \matVt_r
$
where \(\matU\) and \(\matV\) are the matrices obtained be taking the first \(r\) columns of \(\matU\) and \(\matV\) respectively, and \(\Sigma_r = \textnormal{diag}(\sigma_1, ..., \sigma_r)\). Suppose that the rank of \(\matA\) is greater or equal to \(r\), then the rank of \(\matA_r\) is exactly \(r\). 
This truncation of the SVD enjoys many algebraic properties, one such major result is the following Eckart-Young-Mirsky Theorem, stating that the \(r\)-rank truncated SVD of a matrix is in a sense the best approximate of rank lower or equal to \(r\): 
\begin{equation*}
    \matA_r = \argmin_{\tilde{\matA}, \textnormal{rank}(\tilde{\matA}) \leq r} \FNormS{\matA - \tilde{\matA}}
\end{equation*}
This result directly implies~\Cref{app:pca.prop.max.var,app:pca.prop.distortion}  of the PCA that are concerned with the maximization of the variance and minimizing the distortion of the projected configuration. 
Thus, we see that the SVD can be used not only as an alternative construction algorithm for the PCA, but also, due to Eckart-Young Theorem, can serve as the mathematical justification to some of the PCA's key properties.

%% file: child_docs/supp_tcam.tex
A recent work by Kilmer et. al. introduced a tensor version of the Eckart-Young Theorem~\cite{Kilmer}, stating that the truncated tensor SVD is the best low rank approximate within a specified tensor-tensor product framework based on the \(\Mprod\)-product.
The \(\Mprod\)-product operation is defined by an invertible matrix \(\matM\), and the best low rank approximation results from~\cite{Kilmer} were established for matrices \(\matM\) that are nonzero  multiple of a unitary matrix (\(\matM^\T \matM = \matM \matM^{\T} = c^2 \matI_{n}\) for some constant \(c \neq 0\)). In this work, we only consider \(\matM\) that are unitary matrices.

We construct the \tcam on top the tensor SVD introduced by Kilmer et. al., and utilize Kilmer's Eckart-Young-like result to establish tensor analogs of ~\Cref{app:pca.prop.max.var,app:pca.prop.distortion}, in the same way that matrix SVD can be used for deriving the same properties for matrices. 
As a prelude to our discussion, we present some key notions and operations necessary for our derivation. Elaborated introduction and discussion can be found in ~\cite{Kilmer}.

\subsection{Tensor-tensor \texorpdfstring{\(\Mprod\)}--product framework}\label[suppmethod]{subseq:ttprod} 
We begin with some definitions. Let \(\matM\) be an \(n\xx n\) orthogonal matrix (\(\matM \matMt = \mat{I}_n = \matMt \matM \)), and a tensor \(\tA \in \RR^{\mpn}\). We define the {\bf domain transform} specified by \(\matM\) as $\thA \coloneqq \tA \tsM$, where \(\tsM\) denotes the tensor-matrix multiplication of applying \(\matM\) to each of the tensor \(n\) dimensional tubal fibers (\(\tA_{i,j,:}\)). 
The {\bf transpose} of a real \(\mpn\) tensor \(\tA \) with respect to \(\matM\), denoted by \(\tA^{\T}\), is a \(\pmn\) tensor for which $[\widehat{\tA^{\T}}]_{:,:,i} = [\thA^{\T}]_{:,:,i} = {[\thA]_{:,:,i}}^{\T}$.
Given two tensors \(\tA \in \RR^{\mpn}\) and  \(\tB \in \RR^{p \xx l \xx n}\), the facewise tensor-tensor product of \(\tA\) and \(\tB\), denoted by \(\tA \vartriangle \tB\),  is the \(m \xx l \xx n\) tensor for which $[\tA \vartriangle \tB]_{:,:,i} = \tA_{:,:,i} \tB_{:,:,i}$.

We can now define the product The tensor-tensor {\bf \(\Mprod\)-product} of \(\tA \in \RR^{\mpn}\) and  \(\tB \in \RR^{p \xx l \xx n}\) is defined by $\tA \Mprod \tB \coloneqq (\thA \vartriangle \thB) \tsMinv \in \RR^{m \xx l \xx n}$.
A few definitions now naturally follow. 
The  \(p \xx p \xx n\) {\bf identity tensor} with respect to \(\Mprod\), is the tensor \(\tI\) such that for any tensor \(\tE \in \RR^{p \xx p \xx n} \) it holds that $\tI \Mprod \tE = \tE = \tE \Mprod \tI $.
In situations where the dimensions are unclear from context, we use \(\tI_m\) to denote the \(m \xx m \xx n\) identity tensor.
Two tensors \(\tA, \tB \in \RR^{1 \xx m \xx n} \) are called \(\Mprod\){\bf -orthogonal slices} if $\tA^{\T} \Mprod \tB = \mathbf{0}$,  where \(\mathbf{0} \in \RR^{1\xx 1 \xx n} \) is the zero tube fiber, while \(\tQ \in \RR^{m \xx m \xx n}\) is called \(\Mprod\){\bf-unitary} if $\tQ^{\T} \Mprod \tQ = \tI = \tQ \Mprod \tQ^{\T}$.

The following definition is new, and is important for stating the PCA-like properties of \tcam.
\begin{definition}
A tensor \(\tB \in \RR^{p \xx k \xx n}\) is said to be a {\bf pseudo $\Mprod$-unitary tensor} (or {\bf pseudo $\Mprod$-orthogonal}) if \(\tB^{\T} \Mprod \tB\) is f-diagonal (i.e., all frontal slices are diagonal), and all frontal slices of \((\tB^{\T} \Mprod \tB) \tsM \) are diagonal matrices with entries that are either ones or zeros.
\end{definition}

\subsection{The \tsvdm}\label[suppmethod]{app:tsvdm.subsubseq}
\input{./child_docs/supp_tsvdm_tcam}



\subsection{A PCA-like tensor decomposition}\label[suppmethod]{app:tensor.pcalike.subsec}
\input{./child_docs/supp_pcalike_tensor}

\newcommand{\ovec}{\operatorname{vec}}
\subsection{Explicit rank truncated tensors in vector representation}\label[suppmethod]{app:sec.flattened.truncations}
Having established optimality properties with respect to variance and distortion for explicit rank truncated tensors (\Cref{app:tca.prop.max.min.var,app:tca.prop.distortion}), we turn to do the same for vector representation of these truncations.
The vector form of truncated tensors is obtained by mode-1 unfolding.
Let $\tX \in \RR^{1 \xx p \xx n}$, we define $\ovec ( \tX ) \in \RR^{pn}$ as 
\[
\ovec ( \tX ) = [\tX_{1,1,1}, \tX_{1,2,1}, \dots, \tX_{1,p,1},\tX_{1,1,2}, \tX_{1,2,2}, \dots , \tX_{1,p,n} ]~~.
\]
For tensors $\tX \in \RR^{\mpn}$, we slightly abuse the above notation by letting $\ovec(\tX) \in \RR^{m \xx np}$ denote the matrix whose $\ell^{th}$ row is $\ovec ( \tX_{\ell,:,:} )$ (so, even if $m=1$ we will view $\ovec(\tX)$ as a row vector).

Consider the family ${\cal Q}$ of tensor-vector mappings $ \RR^{m \xx p \xx n}$ to $\RR^{m \xx pn}$ whose members are of the form
\begin{equation*}
    \qb (\tX) = \ovec ( (\tX \mm \tB) \tsM )
\end{equation*}
where $\tB \in \RR^{ p \xx p \xx n}$ is a \pmorth tensor. 
We define the {\bf rank} of a member $\qb$ in ${\cal Q}$ as the implicit rank of $\tB$.

In analogy to~\Cref{app:tensor.pcalike.subsec,app:matrix.pca.sec} we present the following properties.
\begin{property}\label{app:vec.prop.max.var}
Suppose that $\matM$ is a unitary matrix. Let $\tA \in \RR^{\mpn}$ be in mean deviation form and $\matY \coloneqq \qb (\tA) \in \RR^{m \xx np}$ for some \pmorth $\tB$. 
As in~\Cref{app:matrix.pca.sec}, we define $\matS_{\matY} \coloneqq (m-1)^{-1 }\matYt \matY $ the sample variance-covariance matrix of $\matY$. 
Then, the rank $q$ member of ${\cal Q}$ for which the variance component $(m-1)^{-1}\trace{\matYt \matY}$ is maximized, is given by $Q_{\tV_{\rrho}}$ where $\tA_{\rrho} = \tU_{\rrho} \mm \tS_{\rrho} \mm \tVt_{\rrho}$ is the multi-rank $\rrho$ truncation $\tA$ under $\mm$ implied by explicit rank $q$ truncation of $\tA$. 
\end{property}

\begin{proof}[Proof of \Cref{app:vec.prop.max.var}]
\label{app:proof.vec.prop.max.var}
    Note that 
    \begin{align*}
        \trace{\matS_{\matY}} &\propto
        \trace{\matY \matYt}\\
        &=\trace{\sum_{h=1}^{np} \matY_{:,h} (\matY_{:,h})^{\T}}  \\
        &=\sum_{h=1}^{np} \trace{\matY_{:,h} (\matY_{:,h})^{\T}}  \\
        &= \sum_{i=1}^{n} \sum_{j=1}^{p} \trace{\thY_{:,j,i} (\thY_{:,j,i})^{\T}} \\ 
        &= \sum_{i=1}^{n} \sum_{j=1}^{p} (\thY_{:,j,i})^{\T} \thY_{:,j,i}   \\ 
        &= \sum_{i=1}^{n} \FNormS{ \thY_{:,:,i} } = \FNormS{ \tY } 
    \end{align*}
    where $\tY = \tA \mm \tB$.
    So we conclude that $\trace{\matS_{\matY}} = (m-1)^{-1} \FNormS{ \tA \mm \tB }$ Now, by~\Cref{app:tca.prop.max.min.var}, we have that the implicit rank $q$ for which this quantity is maximal is $\tV_{\rrho}$.
\end{proof}

\Cref{app:vec.prop.max.var} provides a variance maximization result for flattened explicit rank $q$ truncations that is similar to~\Cref{app:pca.prop.max.var,app:tca.prop.max.min.var} which state variance maximization for traditional rank $q$ truncations of matrices (PCA) and explicit rank-$q$ truncations of tensors respectively.

We proceed with presenting a slightly modified version of ~\Cref{app:tca.prop.distortion,app:pca.prop.distortion}. 
Recall that ~\Cref{app:tca.prop.distortion,app:pca.prop.distortion} discuss the minimization of the Frobenius norm of the difference between configuration matrices 
(or tensors) of the truncated and of the complete representation ; 
\begin{align*}
    \FNormS{\matY  \matYt - \matA \matAt} &,& \FNormS{\tY \mm \tYt - \tA \mm \tAt}
\end{align*}

Formulating analog property for flattened rank $q$ truncated representations, requires a definition of the complete flattened representation of a tensor.
Given $\tA \in \RR^{\mpn}$ in mean deviation form, we let 
\begin{equation}\label{app:tca.flattened.full}
    \matA \coloneqq \qv(\tA) = \ovec(\thZ) 
\end{equation}
denote the complete flattened representation of $\tA$, where $\tA = \tU \mm \tS \mm \tVt = \tZ \mm \tVt $. 
We argue that the definition of the complete flattened representation stated above is rather natural since it holds all information regarding the original tensor $\tA$, and reconstructing $\tA$ from $\matA$ is possible by applying an inverse $\ovec$ operation to $\matA$, followed by a $\xx_{3} \matM^{-1}$ tensor-matrix product, and a right $\mm$ product with $\tVt$.
Additional argument in favor of the choice in~\cref{app:tca.flattened.full}, is that the sample variance component in $\matS_{\matA}$ is identical to that of the original tensor, as could be deduced from the proof of~\Cref{app:vec.prop.max.var}.

Therefore, we will use $\matA$ (\cref{app:tca.flattened.full}) as the reference point for quantifying the distortion in configuration of flattened truncated rank representations of $\tA$.
Note that in contrast to the fact that the expressions $\FNormS{\tY}$ and $\FNormS{\matY}$ are identical (a fact that was used to show~\Cref{app:vec.prop.max.var}), the quantities $\FNormS{\tY \mm \tYt}$ and $\FNormS{\matY \matYt}$ are not equal in general.
This forces a slight modification in the definition of the concept of distortion itself.
Instead of using the traditional formulation $\FNormS{\matY \matYt - \matA \matAt}$ as in~\cite{jolliffe1986principal}, we replace the Frobenius norm with the in nuclear norm.
Given a matrix $\matX$, the nuclear norm of $\matX$, denoted by $\NNorm{\matX}$, is defined as the sum of $\matX$'s singular values.
In the special case where $\matX$ as a positive semi-definite matrix (a case which we briefly describe as $\matX \succeq 0$), we have that $\NNorm{\matX} = \trace{\matX}$.
We now show the following.
\begin{property}\label{app:vec.prop.distortion}
Suppose that $\matM$ is a unitary matrix. Let $\tA \in \RR^{\mpn}$ be in mean deviation form, $\matY \coloneqq \qb (\tA) \in \RR^{m \xx np}$ and $\matA \in \RR^{m \xx pn}$ defined by~\cref{app:tca.flattened.full}.
Then, the rank $q$ member of ${\cal Q}$ for which the distortion $\NNorm{\matY \matYt - \matA \matAt}$ is minimized, is given by $Q_{\tV_{\rrho}}$ where $\tA_{\rrho} = \tU_{\rrho} \mm \tS_{\rrho} \mm \tVt_{\rrho}$ is the multi-rank $\rrho$ truncation $\tA$ under $\mm$ implied by explicit rank $q$ truncation of $\tA$. 
\end{property}

\begin{proof}[Proof of \Cref{app:vec.prop.distortion}]
    Begin by noting that for any $\ell$ and $k$ between $1$ and $m$, we have 
    \begin{align*}
        [\matY \matYt]_{\ell, k} &= \ovec (\thY_{\ell,:,:}) \ovec (\thY_{k,:,:})^{\T} \\
        &= \sum_{i=1}^{n}  \thY_{\ell,:,i} \thY_{k,:,i}^{\T} \\
        &= (\sum_{i=1}^{n}  \thY_{:,:,i} \thY_{:,:,i}^{\T})_{\ell,k} 
    \end{align*}
    and similarly, we get $[\matA \matAt]_{\ell, k} =  (\sum_{i=1}^{n}  \thZ_{:,:,i} \thZ_{:,:,i}^{\T})_{\ell,k} $, 
    thus, $\matY \matYt - \matA \matAt = \sum_{i=1}^{n}  \thY_{:,:,i} \thY_{:,:,i}^{\T}  - \thZ_{:,:,i} \thZ_{:,:,i}^{\T} $.
    \noindent
    We write $\tQ = \tVt \mm \tB$, and note that $\tQ$ is a \pmorth tensor of implicit rank $q$, then
    \begin{align*}
        \thY_{:,:,i} \thY_{:,:,i}^{\T} - \thZ_{:,:,i} \thZ_{:,:,i}^{\T} &= \thZ_{:,:,i} (\thQ_{:,:,i} \thQt_{:,:,i} - \matI)  \thZ_{:,:,i}^{\T}
    \end{align*}
    where $\thQ_{:,:,i} \thQt_{:,:,i} - \matI$ is a $p \xx p $ symmetric matrix. 
    Since $\matI - \thQ_{:,:,i} \thQt_{:,:,i} \succeq 0$
    (it is a projection matrix), we have that $\matA \matAt - \matY \matYt \succeq 0$, and hence 
    \[
    \NNorm{\matA \matAt - \matY \matYt} = \trace{\matA \matAt - \matY \matYt}~.
    \]
    Note that 
    \begin{align*}
        \trace{\matA \matAt - \matY \matYt} &= \trace{ \sum_{i=1}^{n}  \thZ_{:,:,i} \thZ_{:,:,i}^{\T} - \thY_{:,:,i} \thY_{:,:,i}^{\T}} \\
        &= \FNormS{\tZ} - \FNormS{\tY}
    \end{align*}
    a quantity that that is minimized when $\FNormS{\tY}$ is maximized. 
    Since $\FNormS{\tY} $ is proportional to the sample variance component of $\tY$, according to the~\hyperref[app:tca.proof.prop.max.min.var]{proof} of~\Cref{app:tca.prop.max.min.var},  $\NNorm{\matA \matAt - \matY \matYt}^2$ is minimized when $\matY = Q_{\tV_{\rrho}} ( \tA )$, that is, $\tB = \tV_{\rrho}$.
    
\end{proof}

\section{\tcam}
A closer inspection of $\tV_{\rrho}$ reveals that the flattened rank $q$ truncated representation $Q_{\tV_{\rrho}}$, which was shown in~\Cref{app:sec.flattened.truncations} to have optimality guarantees with respect to distortion and variance, can be considerably compressed.
Consider the $i^{th}$ frontal face of $\thV_{\rrho}$, then we have that 
$$
\FNormS{(\thV_{\rrho})_{:, j,i} }  = 
    \begin{cases}
    1 & \exists h \leq q \st (j_h,i_h) = (j,i) \\
    0 & \text{otherwise}
    \end{cases}
$$
where we used the ordering of frontal and lateral indices defined in~\cref{app:tsvdm.ordering.matr}.
Thus, for any tensor $\tX \in \RR^{1 \xx p \xx n }$ we have that $\thX_{:,:,i} (\thV_{\rrho})_{:,:,i} $ can obtain nonzero values only for indices $j=1,\dots,p$ such that $(j,i) = (j_h, i_h)$ for some $h \leq q$, essentially making  $Q_{\tV_{\rrho}}$ a mapping from $\RR^{1 \xx p \xx n}$ to $\RR^{q}$.

Given $\tA \in \RR^{\mpn}$ with explicit rank $q$ truncation $\tA_{\rrho} = \tU_{\rrho} \mm \tS_{\rrho} \mm \tVt_{\rrho}$,  we let $Q:\RR^{1 \xx p \xx n} \to \RR^{q}$ denote the compact representation of $Q_{\tV_{\rrho}}$;
$$
Q(\tX) = [x_1, x_2, ... , x_q]
$$
where $x_h = [(\tX \mm \tV_{\rrho}) \tsM]_{:,j_h,i_h}$. 

For any $\tX \in \RR^{k \xx p \xx n}$, let $\matX \coloneqq Q_{\tV_{\rrho}}(\tX)$ and $\tilde{\matX} \coloneqq Q(\tX)$.
Clearly, $\trace{\tilde{\matX}^{\T} \tilde{\matX}} = \trace{\matXt \matX} $ and also $\tilde{\matX} \tilde{\matX}^{\T} = \matX \matXt  $.
This gives rise to the following:
\begin{definition}[\tcam{}]
Let $\tA \in \RR^{\mpn}$ in mean deviation form. The \tcam{} of $\tA$ consists of the following
\begin{enumerate}
    \item The \tcam{} scores, given by the tensor-to-vector mapping $Q$ specified above. 
    Note that given a new sample $\tX \in \RR^{1 \xx p \xx n}$, the \tcam{} scores of $\tX$ may be easily obtained by applying $Q$ to $\tX$.
    \item The \tcam{} loadings (or coefficients) matrix $\mat{V} \in \RR^{np \xx p} $ with entries $\mat{V}_{h,j} = \thV_{j_h,j,i_h}$, where the index $h$ corresponds to the ordering of the singular values in~\cref{app:tsvdm.ordering.matr}.
\end{enumerate}

The \tcam{} may be truncated at any target number of factors $q$ by taking the first $q$ rows for the scores matrix, and the first $q$ rows in the coefficients matrix.

The scores of the  $q$ truncated \tcam{} enjoy the same properties as $Q_{\tV_{\rrho}} (\tA)$, hence, these scores make variance maximizing vector representation of samples in $\tA$ (\Cref{app:vec.prop.max.var}) while minimizing the distortion with respect to the nuclear norm (\Cref{app:vec.prop.distortion}).
\end{definition}

%% file: child_docs/supp_tsvdm_tcam.tex
With the \(\Mprod\)-product framework set-up, it is now possible to introduce the tensor singular value decomposition (\tsvdm). Let \(\tA \in \RR^{\mpn}\)  be a real tensor, then is possible to write the  full {\bf tubal singular value decomposition} of \(\tA\)  as  $\tA = \tU \Mprod \tS \Mprod \tV^{\T}$, 
where \(\tU, \tV\) are \(m \xx m \xx n\) and \(p \xx p \xx n\) \muni tensors respectively, and \(\tS \in \RR^{\mpn}\) is an {\bf f-diagonal} tensor, that is, a tensor whose frontal slices (\(\tS_{:,:,i}\)) are matrices with zeros outside their main diagonal (see ~\cite{Kilmer} for additional details).
We use the notation $\hsigma_{j}^{(i)}$ do denote the $j^{th}$ largest singular value on the $i^{th}$ lateral face of $\thS$: $\hsigma_{j}^{(i)} \coloneqq \thS_{j,j,i}$.

The \tsvdm construction makes it possible to expand the concept of tensor-rank discussed earlier. The {\bf t-rank} of \(\tA\) is the number of nonezero tubes of \(\tS\): $r = | \left\{ i = 1, \dots, n ~;~ \FNormS{\tS_{i,i,:}} > 0 \right\} |$. 
Additionally, the {\bf multi-rank} of \(\tA\) under \(\Mprod\), denoted by the vector \(\rrho \in \mathbb{N}^{n}\) whose \(i^{th}\) entry is $\rrho_i = \rnk (\thA_{:,:,i})$,
and the {\bf implicit rank} under \(\Mprod\) of a tensor \(\tA\) with multi-rank \(\rrho\) under \(\Mprod\) is $r = \sum_{i=1}^{n} \rrho_i$.


The definitions of tensor t-rank and multi-rank under $\Mprod$ also make it possible to define \tsvdm rank truncation with respect to these ranks. The tensor
$\tA^{(q)} = \tU_{:,1:q, :} \Mprod \tS_{1:q,1:q,:} \Mprod {\tV_{:,1:q,:}}^{\T}$ denotes the {\bf t-rank $q$ truncation} of $\tA$ under $\Mprod$, while multi-rank $\rrho$ truncation of $\tA$ under $\Mprod$ is given by the tensor $\tA_{\rrho}$ for which $\widehat{\tA_{\rrho}}_{:,:,i} = \thU_{:,1:\rrho_i, i}  \thS_{1:\rrho_i,1:\rrho_i,i}  {\thV_{:,1:\rrho_i,i}}^{\T}$. Note that for t-rank truncation the $\tU$ and $\tV$ factors are $\Mprod$-orthogonal, while for multi-rank truncation they are only pseudo $\Mprod$-orthogonal.


The Eckart-Young-like result obtained by Kilmer et. al. states that \(\tA^{(q)}\) and \(\tA_{\rrho}\) are the `best t/multi-rank \(q,\rrho\)' approximations of \(\tA\) respectively, where `best' refers to entrywise squared error, i.e. the Frobenius norm of the error. In other words,  \(\tA^{(qk)}\) and \(\tA_{\rrho}\) are the global minimizers of \(\FNormS{\tA - \tB}\) for  \(\tB\) with t-rank  \(q\) (respectively, multi-rank \(\rrho\)) under \(\Mprod\) of the same dimensions as \(\tA\).

Let \(\tA = \tU \Mprod \tS \Mprod \tV^{\T} \in \RR^{\mpn}\), 
we will use $j_1,\dots, j_{np}$ and $i_1,\dots, i_{np}$ to denote the indexes of the non-zeros of  \(\thS\) ordered in decreasing order. That is
\begin{equation}
    \hsigma_{\ell} \coloneqq \hsigma_{j_{\ell}}^{(i_{\ell})} \label{app:tsvdm.ordering.matr} 
\end{equation}
where $\hsigma_1 \geq \hsigma_2 \geq \dots \geq \hsigma_{np}$.

In this work, we consider truncation with respect to the explicitly given implicit rank under $\Mprod$.
For \(q = 1 , \dots , p n\), the {\bf explicit rank-\(q\) truncation} under \(\Mprod\)  of \(\tA\) is the tensor \(\tA_{\rrho}\) of multi-rank \(\rrho\) under \(\Mprod\)  where 
\begin{equation}\label{app:tsvdm.explicitrank.rho.entry}
    \rrho_i = \max \{ j = 1, \dots ,p ~|~ (j,i) \in \{(j_1, j_1), \dots, (j_q, i_q)\} \} ~~.
\end{equation}
In words, we keep the $q$ top singular values of any frontal slice of $\thS$, and zero out the rest. 
Note that the explicit rank-\(q\) truncation is not always uniquely defined by \(q\) since ties between singular values result in multiple possible choices for the multi-rank \(\rrho\). 
However, all explicit-rank \(q\) truncations are equivalent in that they produce identical reconstruction errors.
Indeed, let \(\tA_{\rrho}\) be a multi-rank \(\rrho\) truncation of \(\tA\), implied by target explicit rank \(q\), then $\FNormS{\tA - \tA_{\rrho}} = \sum_{{\ell = q+1}}^{p \cdot n} \hsigma_{\ell}^{2} $,
meaning that the reconstruction error remains the same for any choice of multi-rank \(\rrho\) in~\cref{app:tsvdm.explicitrank.rho.entry}. 
Furthermore, given the definition of explicit rank truncation, we get the following.
\begin{claim}
Suppose that $\matM$ is a unitary matrix. Let $\tA \in \RR^{\mpn}$, with a full \tsvdm $\tA = \tU \mm \tS \mm \tVt$, and $\rrho = [\rrho_1 , \dots ,\rrho_n ]$ be the multi rank defined by explicit rank-$q$ truncation of $\tA$ in~\cref{app:tsvdm.explicitrank.rho.entry}.
Let $\tA_{\rrho} = \tU_{\rrho} \mm \tS_{\rrho} \mm \tVt_{\rrho} \in \RR^{\mpn}$, the multi-rank $\rrho$ truncation of $\tA$, then $\tA_{\rrho}$ is the best implicit rank-$q$ approximation of $\tA$. 
\end{claim}

\begin{proof}
Let $\pphi = [\pphi_1 ,\dots , \pphi_n]$ be a multi-rank such that $\sum_{i=1}^{n} \pphi_i = q$.
Then $\tA_{\pphi} = \tU_{\pphi} \mm \tS_{\pphi} \mm \tVt_{\pphi} \in \RR^{\mpn}$ is the best multi-rank $\pphi$ approximation of $\tA$. 
Let the tuples $(\tilde{j}_1, \tilde{i}_1), \dots, (\tilde{j}_q, \tilde{i}_q)$ denote an ordering of the singular values $\{ \hsigma_{j}^{(i)} \}$ for $i=1,\dots,n$ and $j = 1,\dots, \pphi_i$,  such that $\hsigma_{\tilde{1}} \geq \hsigma_{\tilde{2}} \geq \dots \geq \hsigma_{\tilde{q}} $ where $\hsigma_{\tilde{\ell}} = \hsigma_{\tilde{j}_{\ell}}^{(\tilde{i}_{\ell})}$.
Then, by construction of $\hsigma_{\ell}$ in~\cref{app:tsvdm.ordering.matr}, we have that $\hsigma_{\ell} \geq \hsigma_{\tilde{\ell}}$ for all $\ell = 1,\dots,q$, hence
$
\sum_{h=1}^{q}  \hsigma_{\tilde{h}}^2 \leq \sum_{h=1}^{q}  \hsigma_{h}^2
$.

Now, we have that 
\begin{align*}
    \FNormS{\tA -\tA_{\pphi}} &= \sum_{i=1}^{n} \sum_{j=1}^{p} (\hsigma_{j}^{(i)})^2  - \sum_{i=1}^{n} \sum_{j=1}^{\pphi_i} (\hsigma_{j}^{(i)})^2 \\
    &= \sum_{h=1}^{pn} \hsigma_{h}^2  - \sum_{h=1}^{q} \hsigma_{\tilde{h}}^2 \\ 
    &\geq \sum_{h=1}^{pn} \hsigma_{h}^2  - \sum_{h=1}^{q} \hsigma_{h}^2 \\
    &= \FNormS{\tA -\tA_{\rrho}}
\end{align*}

Thus, we established that the best implicit rank $q$ approximation of $\tA$, is the tensor $\tA_{\rrho}$ with multi-rank $\rrho$ implied by explicit rank $q$ truncation of $\tA$.
\end{proof}

We also note that for the explicit rank-$q$ truncation of $\tA$, it holds that
\begin{equation}\label{app:tsvdm.erank.k.norm}
    \FNormS{\tA_{\rrho}} = \sum_{h=1}^{q} \hsigma_h^2 
\end{equation}
where $\hsigma_h$ is defined by~\cref{app:tsvdm.ordering.matr}.

%% file: child_docs/supp_pcalike_tensor.tex
The following construction was first presented in~\cite{Hao2013}. Short presentation of the objects we are concerned with are be followed by novel results regarding their algebraic and geometric properties. For extended discussion, see~\cite{Hao2013,Kilmer}

Let \(\tA = \tU \Mprod \tS \Mprod \tV^{\T} \in \RR^{\mpn}\) in mean deviation form (see~\cref{sec:online.methods} section), and consider the following expressions~\cite{Hao2013}:
\begin{equation} \label{app:pca.complete.def.tensor}
    \tZ = \tA \Mprod \tV
\end{equation}
and 
\begin{equation} \label{app:pca.trunc.def.tensor}
    \tZ_{\rrho} \coloneqq \tA \Mprod \tV_{\rrho}
\end{equation}
where the tensor \(\tV_{\rrho} \in \RR^{p \xx p \xx n}\) is obtained from explicit rank $q$ truncated \tsvdm of $\tA$. 
The expressions in \cref{app:pca.trunc.def.tensor} makes use in the \tsvdm to form a construction of high resemblance to the truncated PCA shown in~\cref{app:pca.trunc.def}. 
We also point that $\tV_{\rrho}$ is a member of particular family - it is a \pmorth tensor of implicit rank $q$.

Let $\tB \in \RR^{p \xx p \xx n}$ be a \pmorth tensor of implicit rank $q$, and notice that $(\tBt \Mprod \tB) \tsM$ is an f-diagonal tensor with diagonals consists of zeros and ones, and for which the number of non-zero entries is $q$.
Define
\begin{equation}\label{app:tensor.B.transform}
    \tY = \tA \Mprod \tB  
\end{equation}
In~\Cref{app:pca.prop.max.var} the formulation is concerned with maximizing (or minimizing) the variance component of a random variable \(\y\), that is quantified using the trace of the variance-covariance matrix of that random variable. 
The objective of maximizing (minimizing) the sample variance of $\matY$, that is proportional to $\trace{\matBt \matAt \matA \matB}$, which, by definition of the Frobenius norm, results in $\FNormS{\matA \matB}$.
In absence of analog definitions for variance of a tensor valued random variable and the trace of a tensor, we resort to discuss an algebraically similar trait.

Define the sample variance-covariance  tensor~\cite{Hao2013} for $\tY$ as 
\begin{align} \label{app:tca.reg.energy}
    \tE_{\tY} \coloneqq \tB^{\T} \Mprod \tE \Mprod \tB
\end{align}
where $\tE = (m-1)^{-1}  \tA^{\T} \Mprod \tA \in \RR^{p \xx p \xx n}$ is the sample variance-covariance  tensor for $\tA$ (and $\tZ$).
\begin{property}\label{app:tca.prop.max.min.var}
    Suppose that $\matM$ is a unitary matrix. Let \(\tB \in \RR^{p \xx p \xx n} \) a \pmorth tensor of implicit rank \(q\). 
    Then the quantity \(\FNormS{\tA \Mprod \tB} \) is maximized when \(\tB = \tV_{\rrho}\), where $\tA_{\rrho} = \tU_{\rrho} \mm \tS_{\rrho} \mm \tVt_{\rrho}$ is the explicit rank $q$ truncation of $\tA$. 
\end{property}

\begin{proof}[Proof of \Cref{app:tca.prop.max.min.var}]\label{app:tca.proof.prop.max.min.var}
    We have that $\FNormS{\tA \Mprod \tB} = \FNormS{\thA \vartriangle \thB} = \sum_{i=1}^{n} \trace{\thBt_{:,:,i} \thA^{\T}_{:,:,i} \thA_{:,:,i} \thB_{:,:,i}}$. 
    Recall that \(\tA = \tU \Mprod \tS \Mprod \tVt\), and write $\tB = \tV \Mprod \tC $ where $\tC = \tVt \Mprod \tB$. 
    Note that $\tC$ is also an implicit rank $q$ \pmorth tensor.
    So, $\FNormS{\tA \Mprod \tB} = \sum_{i=1}^{n} \trace{ \thCt_{:,:,i} \thS^{\T}_{:,:,i} \thS_{:,:,i} \thC_{:,:,i}} = \FNormS{\tS \Mprod \tC} $. 
    Thus $\FNormS{\tA \Mprod \tB} = \sum_{i=1}^{n} \sum_{j=1}^{p} (\hsigma_{j}^{(i)})^2  {c}_{j,i}$ where ${c}_{j,i} = \thC_{j,:,i} \thCt_{j,:,i}$, which may be re-ordered according to~\cref{app:tsvdm.ordering.matr} to obtain
    \begin{equation*}
        \FNormS{\tA \Mprod \tB} = \sum_{h=1}^{p n } \hsigma_{h}^{2} c_{j_h, i_h}
    \end{equation*}
    
    %

    Since \(0 \leq c_{j_h, i_h} \leq 1 \) and $\FNormS{\tC}=\sum_{j=1}^{p} {c}_{j,i}=q$, 
    we have that $\sum_{h=1}^{p n } \hsigma_{h}^{2} c_{j_h, i_h} \leq \sum_{h=1}^{q } \hsigma_{h}^{2} $.
    Hence, $\FNormS{\tA \Mprod \tB} \leq \sum_{h=1}^{q} \hsigma_{h}^{2}$, which according to~\cref{app:tsvdm.erank.k.norm},  equals to $\FNormS{\tA_{\rrho}}$ where  $\rrho$ is the multi-rank implied by explicit rank $q$ truncation of $\tA$ (See ~\cref{app:tsvdm.explicitrank.rho.entry}).
    
    Note that taking $\tB = \tV_{\rrho}$ yield  ${c}_{j,i} = ( (\thV_{:,j,i})^{\T} (\thV_{\rrho})_{:,:,i} ) ( (\thV_{:,j,i})^{\T} (\thV_{\rrho})_{:,:,i} )^{\T} $, that in turn results in
    \begin{equation*}
        {c}_{j,i} = \begin{cases}
        1 & \exists h \leq q \st  (j_h, i_h) = (j,i) \\
        0 & \text{otherwise}
        \end{cases}
    \end{equation*}
    thus, the upper bound is achieved for $\FNormS{\tA \Mprod \tV_{\rrho}} = \sum_{h=1}^{q} \hsigma_{h}^{2}$.
   
\end{proof}

The tensor equivalent result for~\Cref{app:pca.prop.distortion} says the minimizer of the distortion under pseudo-orthogonality constraints is again $\tV_{\rrho}$. 
\begin{property}\label{app:tca.prop.distortion}
    Suppose that $\matM$ is a unitary matrix. Let \(\tA \in \RR^{\mpn}\) in mean deviation form and define $\tY = \tA \mm \tB \in \RR^{\mpn}$  where \(\tB \in \RR^{p \xx p \xx n} \) is a \pmorth tensor of implicit rank \(q\).
    Then the distortion formed by \(\tB\), measured as $\FNormS{\tY \Mprod \tYt - \tA \Mprod \tAt}$, is minimized when \(\tB = \tV_{\rrho}\).
\end{property}
\newcommand{\ttheta}{\boldsymbol{\theta}}
\begin{proof}[Proof of \Cref{app:tca.prop.distortion}]
    Note that  \( \thA_{:,:,i}  \thA_{:,:,i}^{\T}\) is symmetric positive semidefinite, thus, \(\thA_{:,:,i}  \thA_{:,:,i}^{\T} = \sum_{j=1}^{\ttheta_{i}}  (\hsigma_j^{(i)})^{2} \u_j^{i} (\u_j^{i})^{\T} \) where  \(\{\u_j^{i}\}_{j=1}^{\ttheta_{i}}\) are unit normalized orthogonal vectors, \((\hsigma_1^{(i)})^{2} \geq  (\hsigma_2^{(i)})^{2} \geq \dots \geq (\hsigma_{\ttheta_{i}}^{(i)})^{2} \) are the eigenvalues of \(\thA_{:,:,i}  \thA_{:,:,i}^{\T}\), and $\ttheta \coloneqq [\ttheta_1 , \dots , \ttheta_{n}]$ is the multi-rank of $\tA$ under $\mm $.
    Similarly, we have that \(\thY_{:,:,i} \thY_{:,:,i}^{\T} = \sum_{j=1}^{\pphi_{i}} \mu_j^{i} \p_j^{i} (\p_j^{i})^{\T} \)  where  \(\{\p_j^{i}\}_{j=1}^{\pphi_{i}}\) are unit normalized orthogonal vectors,  \(\mu_1^{i} \geq  \mu_2^{i} \geq \dots \geq \mu_{\pphi_{i}}^{i}\) are the eigenvalues of \(\thY_{:,:,i} \thY_{:,:,i}^{\T} \), and $\pphi_{i}$ denote the rank of $ \thY_{:,:,i} \thY_{:,:,i}^{\T} $.
    
    Note that  \(\thY_{:,:,i} \thY_{:,:,i}^{\T} = \sum_{j=1}^{\ttheta_{i}} \mu_j^{i}  \u_j^{i} \thB_{:,:,i} \thBt_{:,:,i} (\u_j^{i})^{\T} \), therefore, it must hold that  $\pphi_i \leq \ttheta_i$ for all $i=1, \dots ,n$, meaning that $\sum_{j= 1}^{\pphi_{i}}  (\hsigma_j^{(i)})^{2} \u_j^{i} (\u_j^{i})^{\T} $ is the best rank $\pphi_i$ approximation of $\thA_{:,:,i}  \thA_{:,:,i}^{\T}$.
    As a result, we have $\FNormS{\thY_{:,:,i} \thY_{:,:,i}^{\T} - \thA_{:,:,i}  \thA_{:,:,i}^{\T}} \geq  \sum_{j=\pphi_i + 1}^{\ttheta_{i}}  (\hsigma_j^{(i)})^{4} $ for all $i$. 
    Moreover, since $\tB$ is of implicit rank $q$, it holds that  $\sum_{i=1}^{n} \pphi_i \leq q$.
    Combining the last two, we have
    \begin{align*}
        \FNormS{\tY \mm \tY^{\T} - \tA \mm \tA^{\T}} 
        &\geq  \sum_{i=1}^{n} \sum_{j=1}^{\ttheta_{i}}  (\hsigma_j^{(i)})^{4} - \sum_{i=1}^{n} \sum_{j=  1}^{\pphi_i}  (\hsigma_j^{(i)})^{4} \\
        &\geq \sum_{h=1}^{np} (\hsigma_h)^{4} - \sum_{h=1}^{q}   (\hsigma_{h})^{4}
    \end{align*}
    where \(\hsigma_h \coloneqq \hsigma_{j_h,i_h}\) (See~\cref{app:tsvdm.ordering.matr}). 
    Simple calculations reveals that setting  \(\tB  = \tV_{\rrho}\) results in exactly 
    \begin{equation*}
        \FNormS{\tA \Mprod \tV_{\rrho} \Mprod \tV_{\rrho}^{\T} \Mprod \tAt - \tA \Mprod \tA^{\T}} = \sum_{h=1}^{np} (\hsigma_h)^{4}- \sum_{h=1}^{q} (\hsigma_{h})^{4}
    \end{equation*}
    so we conclude that \(\tV_{\rrho}\) is a global minimizer of the distortion as it reaches the global lower bound.
\end{proof}